\documentclass[11pt]{article}          % twocolumn

\usepackage{amsthm,amsmath,amssymb,amsfonts}

\usepackage{graphicx}

\usepackage{t1enc}

\usepackage{lipsum}
\usepackage{mathtools}
\usepackage{cuted}

\usepackage[margin=1in]{geometry}
\usepackage{ifthen}
\usepackage{graphicx}
\usepackage{hyperref}
\usepackage{float}
\usepackage[font=scriptsize]{caption}
\usepackage[utf8]{inputenc}
\usepackage[T1]{fontenc}
\usepackage[colorinlistoftodos]{todonotes}
\usepackage{algorithmic}
\usepackage{algorithm}
\usepackage{listings}
\usepackage{caption}
\usepackage{subcaption}

\newcommand{\norm}[1]{\left\| #1 \right\|}

\newcommand{\PE}{\mathbb{P}}
\newcommand{\prob}{\mathbb{P}}

\newcommand{\R}{\mathbf{R}}

\newcommand{\NN}{\mathbb{N}}

\newcommand{\vS}{\mathbf{S} }

\newcommand{\N}{\mathcal{N} }

\newcommand{\cM}{{\mathcal M}}

\DeclareMathOperator*{\argmax}{arg\,max}

\usepackage{color}

\usepackage{ulem}
\theoremstyle{plain}

\newtheorem{theorem}{Theorem}
\newtheorem{corollary}{Corollary}
\newtheorem{remark}{Remark}
\newtheorem{proposition}{Proposition}
\newtheorem{lemma}{Lemma}

\theoremstyle{definition}
\newtheorem{assumption}{Assumption}

\theoremstyle{remark}

\newcommand{\formula}[2][nolabel]
{\ifthenelse{\equal{#1}{nolabel}}
 {\begin{align*} #2 \end{align*}}
 {\ifthenelse{\equal{#1}{}}
  {\begin{align} #2 \end{align}}
  {\begin{align} \label{#1} #2 \end{align}}
 }
}

\numberwithin{equation}{section}
\title{VARCLUST: clustering variables using dimensionality reduction}

\author{Piotr Sobczyk\thanks{OLX group in Poland} \and 
        Stanis\l aw Wilczy\'nski\thanks{Microsoft Development Center, Oslo, Norway} \and
        Ma\l gorzata Bogdan\thanks{Institute of Mathematics, University of Wroc\l aw, Poland} 
				\thanks{Department of Statistics, Lund University, Sweden} \and
				Piotr Graczyk\thanks{Laboratoire Angevin de Recherche en Math\'ematiques (LAREMA),  Université d'Angers, France}  \and
        Julie Josse \thanks{INRIA Montpellier and Ecole Polytechnique in Paris, France} \and
				Fabien Panloup\footnotemark[5]\thanks{SIRIC ILIAD, Nantes, Angers, France}  \and
				Val\'erie Seegers \thanks{ICO Angers, France} \and
        Mateusz Staniak \footnotemark[3]\thanks{Department of Computer Science, Northeastern University, USA}}

\providecommand{\keywords}[1]
{
  \small	
  \textbf{\textit{Keywords---}} #1
}

\begin{document}

\maketitle

\begin{abstract}
VARCLUST algorithm is proposed for clustering variables under the assumption that variables in a given cluster are linear combinations of a small number of hidden latent variables, corrupted by the random noise.
The entire clustering task is viewed as the problem of selection of the statistical model, which is defined by the number of clusters, the partition of variables into these clusters and the 'cluster dimensions', i.e. the  vector of dimensions of linear subspaces spanning each of the clusters. The ``optimal'' model is selected using the approximate Bayesian criterion based on the Laplace approximations and
using a non-informative uniform prior on the number of clusters.
To solve the problem of
the search over a huge space of possible models we propose an extension of the {\it ClustOfVar}
algorithm of \cite{COV,COV2}, which was dedicated to subspaces of dimension only 1, and which is similar in structure to the $K$-centroids algorithm.
We provide a complete methodology with theoretical guarantees, extensive numerical experimentations, complete data analyses and implementation. Our algorithm assigns variables to appropriate clusterse based on the consistent Bayesian Information Criterion (BIC), and estimates the dimensionality of each cluster by the PEnalized SEmi-integrated Likelihood Criterion (PESEL) of \cite{sobczyk_pesel}, whose consistency we prove. Additionally, we prove that each iteration of our algorithm leads to an increase of the Laplace approximation to the model
posterior probability and provide the criterion for the estimation of the number of clusters. Numerical comparisons with other algorithms show that VARCLUST may outperform some popular machine learning tools for sparse subspace clustering. We also report the
results of real data analysis including TCGA breast cancer data and meteorological data,
which show that the algorithm can lead to meaningful clustering.
The proposed method is implemented in the publicly available \texttt{R} package varclust.

\end{abstract}

\section{Introduction}
\label{intro}
Due to the rapid development of measurement and computer technologies, large data bases
are nowadays stored and explored in many fields of industry and science.
 This in turn triggered development of new statistical methodology for acquiring information from such large data.

In large data matrices it usually occurs that many of the variables are strongly correlated and in fact convey a similar message. Principal Components Analysis (PCA) \cite{Pearson1901,Hotelling1933,Hotelling1936,pca} is one of the most popular and powerful methods for data compression. This dimensionality reduction method recovers the low dimensional structures spanning the data.  
The mathematical hypothesis which is assumed for this procedure is based upon the belief that
the denoised data matrix is of a low rank, i.e. that the data matrix $X _{n \times p}$ can be represented as
\begin{equation}\label{pesel model}
X = M + \mu + E,
\end{equation}
where $M$ is deterministic,  ${\rm rank}(M)  \ll  \min(n,p)$,
%\ll
  the mean matrix $\mu$ is rank one  and the matrix $E$ represents a centered normal  noise.

Thus, PCA model assumes that all data points come from the
same low dimensional space, which is often unrealistic. 
Fortunately, in many unsupervised learning applications it can be assumed that the data can be well
approximated by a union of lower dimensional manifolds.  One way to analyze such data is to apply 
the nonlinear data projection techniques (see \cite{L_07}). Another approach is to combine local linear models, which 
can often effectively approximate the low dimensional manifolds 
(see e.g., \cite{HS_98}). 
Therefore, in recent years we have witnessed a rapid development of machine learning and statistical technoques for subspace clustering (see \cite{Vidal2011}, \cite{Emmanuel1} and references therein), i.e. for clustering the data into  multiple low dimensional subspaces. As discussed in \cite{Emmanuel1} these techniques have been successfully used in fields as diverse as computer vision (see e.g, \cite{ssc,ssc2}), identification and classification of diseases \cite{scexp} or music analysis \cite{scmusic}, to name just a few.

The most prominent statistical method for subspace clustering is the mixture
of probabilistic principal component analyzers (MPPCA) \cite{mppca}. The statistical model of MPPCA assumes that the rows of the data matrix are independent identically distributed random vectors from the mixture of multivariate gaussian distributions with the low rank covariance matrices. The mixture parameters are usually estimated using the Expectation Maximization algorithm, which is expected to work well when $n>>p$. 

In this paper we propose a more direct approach for subspace clustering, where the fixed effects model (\ref{pesel model}) is applied separately for each cluster.  The statistical model for the whole data base is
determined by the partition of variables into different clusters and the vector of dimensions (ranks of the corresponding $M$ matrices) for
each of the clusters. Our approach allows for creating clusters with the number of variables substantially larger than $n$.
%From a practical (but also statistical)  point of view, the aim of these extensions of PCA-type algorithms is to get a better low dimensional representation of the whole data set, which in turns  should provide some better supervised classification algorithms based on these data. 
%Note that in the biological paradigm, such assumption on the model is natural, due to the existence of well defined \textit{biological pathways}, $i.e.$ of interactions between specific sets of molecules which lead to some changes in the cell. \\
 %\cite{ssc}%n recent years, numerous algorithms have been developed for 
%subspace clustering and especially sparse subspace clustering (see  and the references therein).		 

%extension of this model is
%\noindent In the continuity of \cite{sobczyk_pesel}, where the  BIC-type estimator  of the dimension of a given cluster (named PESEL) was introduced, 
The optimal subspace clustering model is identified through the modified Bayesian Information Criterion, based on the Laplace approximations to the model posterior probability. To solve the problem of
the search over a huge space of possible models we propose in Section \ref{sec:math-model-varclust} a VARCLUST algorithm, which is based on a novel $K$-centroids algorithm, made of two steps based on two different Laplace approximations.
In the first step, given a partition of the variables, we use PESEL \cite{sobczyk_pesel}, a BIC-type estimator,  to estimate the dimensions of each cluster. Cluster centroids are then represented by the respective number of principal components.
% In the first sub-step of the second step, we represent the respective subspaces by providing the appropriate number of first principal components for each  cluster. 
In the second step we perform the partition where the similarity between a given variable and cluster centroid  is measured by the Bayesian Information Criterion in the corresponding multiple regression model. 
\noindent {From a theoretical point of view}, we prove in Section \ref{sec:theory} the consistency of PESEL, $i.e.$ the convergence of the estimator of  the cluster dimension towards its true dimension (see Theorem \ref{thm:pesel}). 
For the VARCLUST itself, we show that our algorithm leads to an increase of the Laplace approximation to the model posterior probability%, up to Laplace approximations used in PESEL and BIC, 
 (see Corollary \ref{cor:main}).  
\noindent {From a numerical point of view, our paper investigates numerous issues in Section  \ref{ss}. %First, in Section \ref{ss}, we focus on simulated examples in order to test the performance of the algorithm. 
The convergence of VARCLUST is empirically checked  and some comparisons with other algorithms are provided showing that the VARCLUST algorithm seems to have the ability to retrieve the true model with an higher probability than other popular sparse subspace clustering procedures. Finally, in Section \ref{sec:realdata}, we consider two important applications to breast cancer and meteorological data. Once again, in this part, the aim is twofold: reduction of dimension but also identification of groups of genes/indicators which seem to take action together.  In Section \ref{sec:package}, the R package \texttt{varclust}  which uses parallel computing for computational efficiency is presented and its main functionalities are detailed.}
\noindent

\section{VARCLUST model} \label{sec:math-model-varclust}

\subsection{A low rank model in each cluster} 

$\;$\\
Let $X_{n\times p}$ be the data matrix with $p$ columns $x_{\bullet j}\in\R^n$, $j\in \{1,\ldots,p\}$. 
The clustering of  $p$ variables $x_{\bullet j}\in\R^n$
 into $K$ clusters consists in considering a column-permuted matrix $X'$ and decomposing
\begin{align}\label{Xcluster}
X'=\left[ X^1 | X^2|\ldots | X^K\right]
\end{align}
such that each bloc $X^i$ has dimension $n\times p_i$, with $\sum_{i=1}^K p_i=p$.
In this paper we apply to each cluster $X^i$
 the model \eqref{pesel model}: 
\begin{equation}\label{cluster model}
X^i = M^i + \mu^i + E^i,
\end{equation}
where $M^i$ is deterministic,  $rank(M^i) = k_i \ll min(n,p_i)$, the mean matrix $\mu^i$ is rank one  and the matrix $E^i$ represents the centered normal  noise $N(0,\sigma^2_i Id)$. 

As explained in \cite{sobczyk_pesel}, the form of  the rank one  matrix $\mu^i$
depends on the relation between $n$ and $p_i$.
When $n>p_i$, the $n$ rows of the  matrix $\mu^i$
are identical, i.e.
$\mu^i= 
\begin{pmatrix}
{\bf r}^i\\
\vdots\\
{\bf r}^i
\end{pmatrix}
$
where ${\bf r}^i=  (\mu_1^i,\ldots, \mu_{p_i}^i)$.
If  $n \le p_i$,  the  $p_i$ columns of the  matrix $\mu^i$
are identical, i.e.
$\mu^i= \begin{pmatrix}
{\bf c}^i \ldots 
{\bf c}^i
\end{pmatrix}
$ with  ${\bf c}^i=(\mu_1^i,\ldots, \mu_{n}^i)^\top$.
We point out that our modeling  allows
some clusters to have $p_i\ge n$, whereas in other clusters $p_i$ maybe smaller than  $n $. This flexibility is one of important advantages of the VARCLUST model.

Next we decompose each matrix $M^i$, for $i=1,\ldots,K$, as a product
\begin{equation}\label{factor model}
M^i=F_{n\times k_i}^i C_{k_i\times p_i}^i
 \end{equation}
The columns of $F_{n\times k_i}$ are linearly
 independent and
will be called factors (by analogy to PCA).\\

This model extends the classical model
\eqref{pesel model} for PCA, which assumes that all variables in the data set can be well approximated by a linear combination of just a few hidden ''factors'', or, in other words, the data belong to a low dimensional linear space. Here we assume that  
 the data comes from a union of such low dimensional subspaces. This means that the variables (columns of the data matrix $X$) can be divided into clusters $X^i$, each of which corresponds to variables from one of the subspaces. Thus, we assume that every variable in a single cluster can be expressed as a linear combination of small number of factors (common for every variable in this cluster) plus some noise. This leads to formulas  
 \eqref{cluster model} and  \eqref{factor model}.
Such a representation is clearly not unique.  
The goal of our analysis is clustering  columns in $X$ and in $M$, such that all coefficients in the matrices $C^1,\ldots,C^K$ are different from zero and $\sum_{i=1}^K k_i$ is minimized. \\

Let us summarize the model that we study. An element of $\mathcal{M}$ is defined by four parameters $(K,\Pi, \vec{k}, \PE_\theta)$
    where:
	 
\begin{itemize}
\item $K$ is the number of clusters and 
$K\le  K_{max}$  for a fixed $K_{max}\ll p$, 
\item $\Pi$ is a $K$-partition of $\{1,\ldots,p\}$ encoding a segmentation of variables (columns of the data matrix $X_{n\times p}$) into clusters $X^i_{n\times p_i}=:X_{\Pi_i}$,
\item $\vec{k} = (k_1, \ldots, k_K) \in\{1,\ldots, d\}^{\otimes K}$, 
  where $d$ is 
	the maximal dimension of
(number of factors in) a cluster.
We choose $d\ll n$ and $ d\ll p$.
%%%%%%%%%%%%%%%%%%%%%%%%%%%%%%%%%%%%%%%%%%%%%%%%%%%%%%%
\item $\PE_{\theta}$ is the probability law of the data specified by the vectors of  parameters $\theta=(\theta_1,\ldots,\theta_K)$,  with $\theta_i$ containing  the factor matrix $F^i$, the coefficient matrix $C^i$,  the 
rank one mean matrix $\mu^i$ and the error variance  $\sigma^2_i $, 
$$\PE_{\theta}(X)=\prod_{i=1}^K \PE\left(X_{\Pi_i}|\theta_i\right)$$
and $\PE\left(X_{\Pi_i}|\theta_i\right)$ is defined as follows: let ${x^i_{\bullet j}}$ be the  $j$-th variable in the $i$-th cluster and let $\mu^i_{\bullet j}$ be the $j$-th column of the
\text{ matrix\ } $\mu^i$. The vectors  ${x^i_{\bullet j}}$,
$j=1,\ldots, p_i$,  are independent conditionally on $\theta_i$ and it holds
\begin{equation}\label{law i}
{x^i_{\bullet j}}|\theta_i={x^i_{\bullet j}}|(F^i, C^i, 
\mu^i , \sigma^2_i) \quad \sim N(F^i  {c^i_{\bullet j}} + \mu^i_{\bullet j}, \sigma^2_iI_n)\;\;.
\end{equation}

\end{itemize}

 Note that according to the model \eqref{law i},
the vectors  ${x^i_{\bullet j}}|\theta_i$, $j=1,\ldots, k_i$,
in the same cluster $X^i$ have  the same covariance matrices
$\sigma^2_iI_n$.

%%%%%%%%%%%%%%%%%%%%%%%%%%%%%%%%%%%%%%%%%%%%%%%%%%%%%%%%%%%%%
%%%%%%%%%%%%%%%%%%%%%%%%%%%%%%%%%%%%%%%%%%%%%%%%%%%%%%%%%%%%%

\subsection{Bayesian approach to subspace clustering}

%Let us now  formulate  the goal of the VARCLUST method, using the Bayes paradigm.
To select a model (number of clusters, variables in a cluster and dimensionality of each cluster), we consider a Bayesian framework. 
 We assume that for any model ${\cM}$ the prior $\pi(\theta)$ is given by
 $$\pi(\theta)=\prod_{i=1}^K \pi(\theta_i)\;\;.$$
Thus, the $\log$-likelihood of the data $X$ given the model $\cM$ is given by
%%%%%%%%%%%%%%%%%%%%%%%%%%%%%%%%%%%%%%%%%%%%%%%%%%%%%%%%
\begin{align} 
\ln\left(\PE(X|\cM)\right) &= \ln\left(\int_{\Theta} \PE(X|\theta) \pi(\theta) d\theta\right) \nonumber\\
&=  \ln \prod_{i=1}^K \left( \int_{\Theta_i}\PE(X^i|\theta_i) \pi(\theta_i) d\theta_i\right) \nonumber\\
&=  \sum_{i=1}^K \ln\left(\int_{\Theta_i} \PE(X^i|\theta_i) \pi(\theta_i) d\theta_i\right) \nonumber\\ 
&= \sum_{i=1}^K \ln\left(\PE(X^i|\cM_i)\right), \label{prob}
\end{align}
where $\cM_i$ is the model for the $i$-th cluster $X^i$ specified by (\ref{factor model}) {and \eqref{law i}}.

%or, equivalently, its logarithm (the $\log$-likelihood  function $\log f_\mathcal{M}({X}^{obs})$.

%%%%%%%%%%%%%%%%%%%%%%%%%%%%%%%%%%%%%%%%%%%%%%%%%%%%%%%%%%%%%%

In our approach we propose an informative prior distribution on $\cM$. The reason is that in our case we have, for given $K$, roughly  $K^p$ different segmentations, where $p$ is the number of variables. Moreover, given  a maximal cluster dimension $d=d_{max}$, there are $d^K$  different selections of cluster dimensions. Thus, given $K$, there are approximately $K^p d^K$ different models to compare. This number quickly increases with $K$ and assuming that all models are equally likely we would in fact use a prior on the number of clusters $K$, which would strongly prefer large values of $K$. Similar problems were already observed when using BIC to select the multiple regression model based on a data base with many potential predictors. In \cite{mBIC_first_paper} this problem was solved by using the modified version of the Bayes Information Criterion (mBIC) with the informative sparsity inducing prior on $\cM$. Here we apply the same idea and use an approximately uniform prior on $K$ from the set $K\in \{1,\ldots, K_{max}\}$, which, for every model ${\cM}$ with the number of clusters $K$, takes the form:
%%%%%%%%%%%%%%%%%%%%%%%%%%%%%%%%%%%%%%%%%%%%%%%%%%%%%%%
\begin{align}
   \pi(\cM)&=\frac{C} {K^p d^K} \nonumber \\
   \ln(\pi(\cM)) &= -p\ln(K) - K\ln(d)+C\ , \label{prior}
\end{align}
where $C$ is a proportionality constant, that does not depend on the model under consideration.
Using the above formulas and the Bayes formula, we obtain the following Bayesian criterion for the model selection:
pick the model (partition $\Pi$ and cluster dimensions $\vec{k}$) such that 
\begin{align}
\ln(\PE(\cM|X)) &= \ln(\PE(X|\cM)) + \ln(\pi(\cM)) -\ln (\PE(X))
\label{bayes} \nonumber \\ 
&= \sum_{i=1}^K  \ln \PE(X^i|\cM_i) -p\ln(K) \\
& \qquad - K\ln(d) + C - \ln \PE(X) \nonumber \;\;.
  \end{align}
obtains a maximal value. 
Since $\PE(X)$ is the same for all considered models this amounts to selecting the model, which optimizes the following criterion
\begin{equation}\label{crit1}
C(\cM|X)= \sum_{i=1}^K 
\ln \PE(X^i|\cM_i) -p\ln(K) - K\ln(d)\;.
\end{equation}
The only quantity left to calculate in the above equation is $ \PE({X}^i|\cM_i)$.

\section{VARCLUST method} \label{sec:algo}
\subsection{Selecting the rank in each cluster with the PESEL method}\label{sec:pesel}
Before presenting the VARCLUST method,
let us present shortly the PESEL method,
introduced in  \cite{sobczyk_pesel} designed to estimate the number of principal components in PCA.
It will be used in the first step of the VARCLUST. 

As explained in  Section \ref{sec:math-model-varclust}
(cf. \eqref{law i}), our model for one cluster can be described by its set of parameters (for simplicity we omit the index of the cluster) $\theta: F \in \mathbb{R}^{n\times k} , c_{1}, \ldots, c_{p},$ where $c_i \in \mathbb{R}^{k\times 1} \text{ (vectors of coefficients)}, \sigma^2$ and $\mu$. In order to choose the best model we have to consider models with different dimensions, {\it i.e.} different values of $k$. The penalized semi-integrated likelihood (PESEL, \cite{sobczyk_pesel}) criterion is based on the Laplace approximation to the semi-integrated likelihood and in this way it shares some similarities with BIC. The general formulation of PESEL allows for using different prior distributions on $F$ (when $n>p$) or $C$ (when $p>n$). The basic version of PESEL uses the standard gaussian prior for the elements of $F$ or $C$ and has the following formulation, depending on the relation between $n$ and $p$. 

We denote by $(\lambda_j)_{j=1,\ldots, p}$
the non-increasing sequence of eigenvalues of the sample  covariance  matrix $S_n$. When $n \leq p$ we use the following form of the PESEL 
{\small
\begin{align}
&\ln(\PE(X^i|\cM_i)) \approx PESEL(p,k,n)=   \nonumber\\  
&-\frac{p}{2} \left[\sum^k_{j=1} \ln(\lambda_j) + (n-k)\ln\left(\frac{1}{n-k}\sum_{j=k+1}^n \lambda_j\right) + n\ln(2\pi)+n\right] \nonumber \\
& \qquad - \ln(p)\frac{nk- \frac{k(k+1)}{2} + k + n + 1}{2}
\label{bic}
\end{align}
}
 and when $n >p$ we use the form
{\small
\begin{align}
 &\ln(\PE(X^i|\cM_i))  \approx PESEL(n,k,p)= 
 \nonumber \\  &-\frac{n}{2}\left[\sum^k_{j=1} \ln(\lambda_j) + (p-k)\ln\left(\frac{1}{p-k}\sum_{j=k+1}^p \lambda_j\right) + p\ln(2\pi)+p\right] \nonumber \\
 & \qquad - \ln(n)\frac{pk- \frac{k(k+1)}{2} + k + p + 1}{2}\;\;. \label{bic2}
\end{align}
}
%The very detailed derivation of the PESEL criterion can be found in Appendix B of Stachu's paper.
The function of the eigenvalues $\lambda_j$ of $S_n$ appearing in
\eqref{bic},\eqref{bic2}  and  approximating the $\log$-likelihood  $\PE(X^i|\cM_i)$ is called a {\it PESEL function}.
The criterion consists in choosing the value of $k$ maximizing the PESEL function.

When $n>p$, the above form of PESEL coincides with BIC
in Probabilistic PCA (see \cite{Minka00automaticchoice}) or the spiked covariance structure model. These models assume that the rows of the $X$ matrix are i.i.d. random vectors. Consistency  results for PESEL under these  probabilistic
 assumptions can be found in
\cite{Bai2018}. 

In Section \ref{ssection:pesel} we will prove consistency of PESEL under a much  more general fixed effects model \eqref{eq:model_for_consistency}, which does not assume  the equality of laws of rows in $X$.

%Our model \eqref{eq:model_for_consistency}  does not assume , thus we have many more parameters in the model.  
%%%%%%%%%%%%%%%%%%%%%%%%%%%%%%%%%%%%%%%%%%%%%%%%%%%%

\subsection{Membership of a variable in a cluster with the BIC criterion}

To measure the similarity between $l^{th}$ variable and a subspace corresponding to $i^{th}$ cluster  we use the Bayesian Information Criterion.
Since the model
\eqref{law i} assumes that all elements of the error matrix $E^i$ have the same variance,
we can estimate $\sigma_i^2$ by  MLE 
$$
 \hat\sigma^2_i=\frac {\sum_{\ell\in \Pi_i} \|{x}_{ \bullet \ell}-P_i({x}_{ \bullet \ell})\|^2}{np_i}\;\;,
$$
where $P_i({x}_{ \bullet \ell})$ denotes the orthogonal projection of the column ${x}_{ \bullet \ell}$ on the linear space corresponding to the $i^{th}$ cluster
and next use BIC of the form
 \begin{equation}\label{bic:sigmaconst}
 BIC(l,i)= \frac{1}{2}\left(- \frac {\|{x}_{ \bullet \ell}-P_i({x}_{ \bullet \ell})\|^2}
{{\hat\sigma^2_i}} - {\ln n} \, k_i \right)\;\;.
\end{equation}

As an alternative, one can consider a standard multiple regression BIC, which allows for different variances in different columns of $E^i$:
 \begin{equation}\label{bic:rss}
BIC(l,i) = -n\ln\left(\frac{RSS_{li}}{n}\right) - k_i \ln(n),
\end{equation} 
 where $RSS_{li}$ is the residual sum of squares from regression of ${x}_{ \bullet \ell}$ on the basis vectors spanning $i^{th}$ cluster.

%%%%%%%%%%%%%%%%%%%%%%%%%%%%%%%%%%%%%%%%%%%%%%%%%

\subsection{VARCLUST algorithm}

\paragraph{Initialization and the first step of VARCLUST}
% {\bf The initialization by  choice of $\Pi^0$.}
%\\
%Deterministic initialization: take  "nearly equidistributed" $p^0_i=[p/K],
%i=1,\ldots,K-1, p^0_K=p-(K-1)[p/K]$. For example, when
%$p=50, K=3$, we have $[p/K]=16$ and $p^0=(16,16,18)$. 
%$\Pi^0$ consists in grouping together the consecutive $16,16,18$ columns of $X$.\\
Choose randomly a $K$-partition of $p=p^0_1+\ldots+p^0_K$
and group randomly $p^0_1,\ldots,p^0_K$ columns of $X$ to form $\Pi^0$.\\
Then, VARCLUST proceeds as follows:
\begin{equation}\label{firststep}
\Pi^0\ \to\ (\Pi^0,k^0)\ \to\ (\Pi^1,k^0)\;\;,
\end{equation}
where $k^0$ is computed by using  PESEL  $K$ times, separately
 to each matrix $X^i_0$, $i=1,\ldots, K$. Next, 
for each 	matrix $X^i_0$, PCA is applied to estimate $k^0_i$ principal factors $F^1_i$, which represent the basis spanning the subspace supporting $i^{th}$ cluster and the center of the clusters. 
The next partition $\Pi^1$ is obtained by using $BIC(l,i)$ as a measure of similarity between $l^{th}$ variable and $i^{th}$ cluster to allocate each variable to its closest cluster. 
After the first step of VARCLUST, we get the couple: the partition and the vector of cluster dimensions $(\Pi^1,k^0)$.\\[2mm]
%%%%%%%%%%%%%%%%%%%%%%%%%%%%
Other schemes of initialization can be consider such as a one-dimensional initialization. 
%Consider the canonical basis $(e_j)$ of $\R^p$. 
Choose randomly $K$ variables which will play the role of one dimensional centers of $K$ clusters . 
Distribute, by minimizing BIC, 
the $p$  columns of $X$ to form the first partition $\Pi^1$. In this way, after the first step of VARCLUST we again get $(\Pi^1, k^0)$, where $k^0$ is the $K$ dimensional all ones vector. 

\paragraph{ Step $m+1$ of VARCLUST} 
 In the sequel  we continue by first using PESEL to calculate a new vector of dimensions and next PCA and BIC to obtain the next partition:

 $$
(\Pi^{m},k^{m-1})\ \to\ (\Pi^{m},k^m)\ \to\ (\Pi^{m+1},k^m).
$$

\section{Theoritical guarentees}\label{sec:theory}

In this Section we prove the consistency of PESEL and show that each iteration of VARCLUST  asymptotically leads to an increase of the objective function (\ref{crit1}).

\subsection{Consistency of PESEL}\label{ssection:pesel}

In this section we prove that PESEL consistently estimates the rank of the denoised data matrix. The consistency holds when $n$ or $p$ diverges to infinity, while the other dimension remains constant. This result can be applied separately to each cluster $X^i$, $i\in\{1,\ldots,K\}$, of the full data matrix. 

First, we prove the consistency of PESEL (Section \ref{sec:pesel}) when $p$ is fixed as $n\rightarrow \infty$.

\begin{assumption}\label{as1}

Assume that the data matrix $X$ is generated according to the following probabilistic model  :\\
\begin{equation}\label{eq:model_for_consistency}
X_{n\times p}   = M_{n\times p} + {\mu}_{n \times p} + E_{n\times p},
\end{equation}
where 

	\begin{itemize}
	\item for each $n\in\NN$, matrices $ M_{n\times p}
	%=\vM(n)
	$ and $ {\mu}_{n \times p}$ are deterministic 
	\item ${\mu}_{n \times p}$ is a rank-one matrix in which all rows are identical, i.e. it represents average variable effect.
	\item the matrix $M_{n\times p}$ is  centered: $\sum_{i=1}^n M_{ij} = 0$ and $\text{rank}(M_{n\times p})=k_0$  for all $n\ge k_0$ 
	\item  the elements of matrix  $M_{n \times p}$  are bounded: $\sup_{n, i \in (1, \dots, n), j \in (1, \dots, p)} |M_{ij}| < \infty$
	\item  there exists the limit:
	$\lim_{n\to\infty}\frac{1}{n} M_{n\times p}^T M_{n \times p} =L$ and, for all $n$
%%%%%%%%%%%%%%%%%%%%%%%%%%%%%%%%%%%%%%
\begin{equation} \label{eq:pesel_consistency_M_assumption}
\left|\cfrac1n M_{n\times p}^T M_{n \times p} - L\right|< C \frac{\sqrt{2 \ln \ln n}}{\sqrt{n}}\;\;,
\end{equation}
where $C$ is some positive constant and $L=
	U D_{p \times p} U^T$  with
$$D_{p \times p}  =
	\begin{pmatrix}
diag[\gamma_i ]_{i=1}^{k_0} & 0\\ 0 &
diag[ 0   ] 
\end{pmatrix}
$$
with non-increasing $\gamma_i>0$ and $U^T U=Id_{p\times p}$. 
	\item the noise matrix $ E_{n\times p}$ consists of i.i.d.
	terms $e_{ij} \sim \N(0,\sigma^2)$.
	\end{itemize}
\end{assumption}

\begin{theorem}[Consistency of PESEL]\label{thm:pesel}

Assume that the data matrix  $X_{n\times p}$ satisfies the Assumption \ref{as1}. 
Let $\hat k_0(n) $ be the $\text{PESEL}(p,k,n)$ estimator of the rank of $M$.

Then, for $p$ fixed, it holds
$$
\prob(\exists {n_0} \ \forall {n > n_0}\quad \:  \hat k_0(n)= k_0) = 1.$$

\end{theorem}

\noindent {\it Scheme of the Proof.}

Let us consider the sample covariance matrix
$$S_n = \cfrac{(X - \bar{X}) ^T (X- \bar{X})}{n}.$$ 
\\
and the population covariance matrix
$ \Sigma_n = E \left( S_n \right)$.
The idea of the proof is the following. 

%As  the value of the PESEL criterion, denoted by $\text{PESEL}_n(\vX, k)$ for a given number of principal components $k$ depends \textbf{only} on a sample covariance matrix $\vS_n$ and not on the data $\vX$ directly, we shall treat it as a function of sample covariance matrix.

Let us denote by $F(n,k)$ the  PESEL function in the case when $n>p$.
By \eqref{bic2}, we have
\begin{align*}
%\label{BIC_model_PPCA}
& F(n,k)=  \\
&-\frac{n}{2}
\Bigl[\sum^k_{j=1} \ln(\lambda_j) + (p-k)\ln\left(\frac{1}{p-k}\sum_{j=k+1}^p \lambda_j\right) \Bigr. \\
& \qquad + \Bigl.  p\ln(2\pi) + p \Bigr] \\
& \qquad - \ln(n)\frac{pk- \frac{k(k+1)}{2} + k + p + 1}{2}
\end{align*}
\\

The proof comprises two steps. First,  we quantify the difference between eigenvalues of matrices $S_n$, $\Sigma_n$ and $L$. We prove it to be bounded by the matrix norm of their difference, which goes to 0 at the pace $\frac{\sqrt{\ln \ln n}}{\sqrt{n}}$ as $n$ grows to infinity, because of the law of iterated logarith (LIL). We use the most general form of LIL from   \cite{petrov1995limit}. 
Secondly, we use the results from the first step to prove that for sufficiently large $n$
the PESEL function $F(n,k)$ 
 is increasing for $k < k_0$ and decreasing for $k>k_0$.
 To do this, the crucial Lemma \ref{lemmaLILeigen}    is proven and used.
The detailed proof is given in Appendix \ref{sec:appendix}.

Since the version of PESEL for $p>>n$, PESEL$(n,k,p)$, is obtained simply by applying  PESEL(p,k,n) to the transposition of $X$, Theorem \ref{thm:pesel} implies the consistency of PESEL also in the situation when $n$ is fixed, $p\rightarrow \infty$ and the transposition of $X$ satisfies the Assumption \ref{as1}.

\begin{corollary}
Assume that the transposition of the data matrix $X_{n\times p}$ satisfies the Assumption \ref{as1}. 
Let $\hat k_0(n) $ be the $\text{PESEL}(n,k,p)$ estimator of the rank of $M$.

Then, for $n$ fixed, it holds
$$
\prob(\exists {p_0} \ \forall {p > p_0}\quad \:  \hat k_0(p)= k_0) = 1.$$

\end{corollary}

%%%%%%%%%%%%%%%%%%%%%%%%%%%%%%%%%%%%%%%%%%%%%%%%%%%%%%%%%%%%%%

\begin{remark} The above results assume that $p$ or $n$ is fixed. We believe that they hold also in the situation when $\frac{n}{p}\rightarrow \infty$ or vice versa. The mathematical proof of this conjecture is an interesting topic for a further research. These theoretical results justify the application of PESEL when $n>>p$ or $p>>n$. Moreover, simulation results reported in \cite{sobczyk_pesel} illustrate good properties of PESEL also when $p\sim n$. The theoretical analysis of the properties of PESEL when $p/n\to C \neq 0$ remains an interesting topic for further research.
\end{remark}
%%%%%%%%%%%%%%%%%%%%%%%%%%%%%%%%%%%%%%%%%%%%%%%%%%
\subsection{ Convergence of VARCLUST}

As noted above in \eqref{crit1}, the main goal of VARCLUST is identifying the model $\cM$ which maximizes, for a given dataset $X$,
$$\ln(\PE(\cM|X)) = \sum_{i=1}^K 
\ln \PE(X^i|k_i) + \ln(\pi(\cM))\;\;,$$
where $\ln(\pi(\cM))$ depends only on the number of clusters $K$ and the maximal allowable dimension of each cluster $d$.

Since, given the number of clusters $K$, the VARCLUST model is specified by the vector of cluster dimensions  $k=(k_1,\ldots, k_K)$  and a partition $\Pi=(\Pi_1,\ldots,\Pi_K)$ of $p$ variables into these $K$ clusters,
 our task reduces to identifying the model for which 
the following objective function
\begin{equation}\label{objective}
\varphi(\Pi, k)\ :=\sum_{i=1}^K 
\ln \PE(X^i|k_i)\;\;,
\end{equation}
obtains a maximum.

Below we will discuss consecutive steps of the VARCLUST Algorithm  with respect to the optimization of (\ref{objective}). 
Recall that the $m+1$ step of VARCLUST  is
 $$
(\Pi^{m},k^{m-1})\ \to\ (\Pi^{m},k^m)\ \to\ (\Pi^{m+1},k^m),
$$
where we first  use PESEL to estimate the dimension and next PCA to compute the factors and BIC to allocate variables to a cluster.

\begin{enumerate}

\item {\bf PESEL step: choice of cluster dimensions, for a fixed partition of $X$.}

First, observe that the dimension of $i^{th}$ cluster in the next $(m+1)^{th}$ step of VARCLUST is obtained as
$$k^m_i = \argmax_{k_i\in\{1,\ldots,d\}} PESEL(X^i|k_i)\;\;.$$
Thus,  denoting by $PESEL$ the PESEL function from \eqref{bic} and \eqref{bic2},
\begin{equation*}\label{peselpart1}
\sum_{i=1}^K   PESEL(X^i_m|k^m_i)\geq \sum_{i=1}^K   PESEL(X^i_m|k^{m-1}_i)\;\;.
\end{equation*}

Now, observe that under the standard regularity conditions  for the Laplace approximation (see e.g. \cite{BIC})
$$\ln \PE(X^i|k_i)=PESEL(X^i|k_i)+O_n(1)\;\;$$
 when $n\rightarrow \infty$ and $p_i$ is fixed and
$$\ln \PE(X^i|k_i)=PESEL(X^i|k_i)+O_{p_i}(1)\;\;$$
 when $p_i \rightarrow \infty$ and $n$ is fixed
(see \cite{sobczyk_pesel}).
Thus,
$$\varphi(\Pi, k)=\sum_{i=1}^K PESEL(X^i|k_i) + R\;\;,$$
where the ratio of $R$ over $\sum_{i=1}^K PESEL(X^i|k_i)$ converges to zero in probability, under our asymptotic assumptions.

Therefore, the first step of VARCLUST leads to an increase of 
$\varphi(\Pi, k)$ up to Laplace approximation, i.e. with a large probability when for all $i \in \{1,\ldots,K\}$, $n>>p_i$ or $p_i>>n$.

%%%%%%%%%%%%%%%%%%%%%%%%%%%%%%%%%%%%%%%%%%%%%%%%%%%%%%%%%

\item  {\bf PCA and Partition step: choice of a partition, with cluster dimensions $k^m_i$ fixed.}

 In the second step of the $m+1$-st iteration of VARCLUST,  the cluster dimensions $k^m_i$ are fixed, PCA is used to compute  the cluster centers $F^i$
and  the columns of $X$ are partitioned to different clusters by minimizing the BIC distance from $F^i$.

Below we assume that the priors $\pi_C(dC)$ and  $\pi(d\sigma)$ satisfy classical regularity conditions for Laplace approximation (\cite{BIC}).
Now, let us define the $k_i^m$--dimensional linear space through the set of respective directions $F^i=(F^i_1,\ldots,F^{i}_{k_i})$ with, as a  natural prior,  the uniform distribution $\pi_F$ on the compact Grassman manifold $F$ of free $k_i$-systems of $\mathbb{R}^n$. Moreover, we assume that the respective columns of coefficients $C^i=(C^i_1,\ldots,C^i_{k_i})$ are independent with a prior distribution $\pi_C$ on $\mathbb{R}^{p}$.

It holds  
\begin{align*}
\log \PE(X^i|k_i)&=\log \int_{F\times\sigma} \int_C \PE(X^i|F^i,C^i,\sigma_i) \\
& \qquad \qquad \qquad \pi(dC^i) \pi(d\sigma_i) \pi_F (dF^i)\\
&=\log \int_{F\times\sigma} \prod_{\ell\in \Pi^i} \int \PE(X_{\bullet \ell}|F^i,C_{\bullet \ell}) \\
& \qquad \qquad \qquad \pi_C (dC_{\bullet \ell})  \pi(d\sigma_i) \pi_F (dF^i).
\end{align*}

When $n\gg k_i$, a Laplace-approximation argument  leads to
$$\int \PE(X_{\bullet \ell}|F^i,C_{\bullet \ell}) \pi_C (dC_{\bullet \ell})\approx e^{{\rm BIC}_\ell |F_i, \sigma_i}$$
where
 $${\rm BIC}_\ell| F_i,\sigma_i=\frac{1}{2}\left(-\frac {\|{x}_{ \bullet \ell}-P_i({x}_{ \bullet \ell})\|^2}
{{\sigma_i^2}} -k_i {\ln n}\right).$$
 
Thus, thanks  to the Laplace approximation above, 
\begin{align}\label{integral0}
\log & \PE(X^i|k_i) \nonumber \\
& \approx \log \int_{F^i\times\sigma_i} e^{\sum_{\ell\in \Pi^i}{\rm BIC}_\ell |F_i, \sigma_i}\pi(d\sigma_i) \pi_F (dF^i)\;\;
\end{align}

and
\begin{align}\label{integral}
\sum_{i=1}^K & \log \PE(X^i|k_i) \nonumber \\
& \approx \log \int_{F\times\sigma} e^{\sum_{i=1}^K \sum_{\ell\in \Pi^i}{\rm BIC}_\ell |F_i, \sigma_i}\pi(d\sigma) \pi_F (dF)\;\;.
\end{align}

%= ln \int_F   e^{\sum BIC_\ell|F} dF
%&=\int_F e^{n L(F)}\pi(F) dF
%\end{align*}
%where $L(F)=\int_C \PE(x_{\bullet \ell}|F,C_{\bullet \ell}) \pi(C_{\bullet \ell}) d C_{\bullet \ell})$.
%& \approx \log \int_{F}  \prod_{\ell \in \Pi^i}  e^{BIC_\ell |F_i} dF_i= \log \int_{F^i}    e^{\sum BIC_l|F_i} dF_i\\
%&\approx \log \int_{F}  \prod_{\ell \in \Pi^i} e^{BIC_\ell |E_i} dF_i= ln \int_F   e^{\sum BIC_\ell|F} dF
%\end{align*}
%where ${\rm BIC}_\ell=\log \PE(x_{\bullet l}|\hat C^i, F^i) - 1/2 k_i \log n$. 

  Now, by Laplace approximation, when $p_i>>k_i$, the right-hand side of  \eqref{integral} can be approximated by
	\begin{equation}\label{newLaplace}
	\psi(\Pi|k) - \sum_{i=1}^K\frac{\dim F_i+1}2 \ln n,
	\end{equation}
	where we denote

\begin{align}
\psi(\Pi|k) &=max_{( F, \sigma)} \xi (\Pi, F, \sigma|k),
\label{function1}
\\
\label{function}
\xi(\Pi, F, \sigma|k)&=\sum_{i=1}^K \sum_{\ell \in \Pi^i} \left(-\frac {\|{x}_{ \bullet \ell}-P_i({x}_{ \bullet \ell})\|^2} 
{{\sigma_i^2}} - {\ln n}\, k_i\right)\;\;.
\end{align}

%then \mb{the majority of the mass of the integral (\ref{integral}) is closely concentrated around the maximum of $\psi(\Pi,E, \sigma)$.} 

Now, the term $\ln n \sum_{i=1}^K\frac{\dim F_i+1}2 $
in \eqref{newLaplace} is the same for each $\Pi$, so 
increasing  \eqref{integral} is equivalent to increasing
$\psi(\Pi|k)$.

 Now, due to the well known Eckhart-Young theorem, for each $i\in\{1,\ldots,K\}$, the first $k_i$ principal components of $X^i$ form the basis for the linear space ''closest'' to $X^i$, i.e.  the PCA  part of VARCLUST allows to obtain $F^m$ and $\sigma^m$, such that
$$(F^m, \sigma^m|\Pi^m,k^m)=argmax_{F,\sigma} \xi(\Pi^m, F,\sigma|k^m)\;.$$
Thus $\psi(\Pi^m|k^m)= \xi(\Pi^m, F^m,\sigma^m|k^m)$.

 Finally, in the Partition (BIC) step  of the algorithm the partition
$\Pi^{m+1}$ is selected such that
 $$\Pi^{m+1}|E^m, \sigma^m,k^m=argmax_{\Pi} \xi(\Pi, E^m,\sigma^m|k^m)\;\;.$$

In the result it holds that
 $$\psi(\Pi^{m+1}|k^m)\geq \psi(\Pi^{m}|k^m)\;\;$$
and consequently,
\begin{equation*}\label{hope}
\varphi(\Pi^{m+1}, k^{m})\geq \varphi(\Pi^{m}, k^{m}) \;\;,
\end{equation*}

with a large probability if only $k_i<<min (n,p_i)$ for all $i\in \{1,\ldots,K\}$.
\end{enumerate}

The combination of results for both steps of the algorithm 
implies 
\begin{corollary}\label{cor:main}
In the VARCLUST algorithm,
the objective function $\varphi(\Pi^{m+1}, k^{m})$ increases with $m$ with a large probability if for all $i\in\{1,\ldots,K\}$, $k_i<<min (n,p_i)$ and one of the following two conditions holds: $n>>p_i$ or $p_i>>n$ . 
\end{corollary}

\begin{remark}

The above reasoning illustrates that both steps of VARCLUST asymptotically lead to an increase of the same objective function. The formula (\ref{function}) suggests that this function is bounded, which implies that VARCLUST converges with a large probability. In Figure    \ref{fig:iter} we illustrate the convergence of VARCLUST based on the more general version of BIC (\ref{bic:rss}) and a rather systematic increase of the mBIC approximation to the model posterior probability 
\begin{align*}
mBIC&(K,\Pi,k) \\
&= \sum_{i=1}^K   PESEL(X^i_m|k^m_i) - p \ln K - K\ln d
\end{align*}

in consecutive iterations of the algorithm.

\end{remark}

\section{ Simulation study}\label{ss}
In this section, we present the results of simulation study, in which we compare VARCLUST with other methods of variable clustering. To assess the performance of the procedures we measure their effectiveness and execution time. We also use VARCLUST to estimate the number of clusters in the data set. In all simulations we use VARCLUST based on the more general version of BIC (\ref{bic:rss}).
\subsection{Clustering methods}

In our simulation study we compare the following methods:
\begin{enumerate}
\item Sparse Subspace Clustering (SSC, \cite{ssc, Emmanuel1})
\item Low Rank Subspace Clustering (LRSC, \cite{lrsc})
\item VARCLUST with multiple random initializations.
In the final step, the initialization with the highest mBIC is chosen.
\item VARCLUST with initialization by the result of SSC (VARCLUST$_{aSSC}$)
\item ClustOfVar (COV, \cite{COV}, \cite{COV3})
\end{enumerate} 

The first two methods are based on spectral clustering and detailed description can be found in the given references. Specifically, Sparse Subspace Clustering comes with strong theoretical guarantees, proved in \cite{Emmanuel1}. For the third considered procedure we use the one-dimensional random initialization. This means that we sample without replacement $K$ variables which are used as one dimensional centers of $K$ clusters. The fourth method takes advantage of the possibility to provide the initial segmentation before the start of the VARCLUST procedure. It accelerates the method, because then there is no need to run it many times with different initializations. We build the centers by using the second step of VARCLUST (PESEL and PCA) for a given segmentation. In this case we use the assignment of the variables returned by SSC. Finally, we compare mentioned procedures with COV, which VARCLUST is an extended version of. COV also exploits k-means method. Initial clusters' centers are chosen uniformly at random from the data. Unlike in VARCLUST the center of a cluster is always one variable. The similarity measure is squared Pearson correlation coefficient. After assignment of variables, for every cluster PCA is performed to find the first principal component and make it a new cluster center. VARCLUST aims at overcoming the weaknesses of COV. Rarely in applications the subspace is generated by only one factor and by estimating the dimensionality of each cluster VARCLUST can better reflect the true underlying structure. 

\subsection{Synthetic data generation}
To generate synthetic data to compare the methods from the previous section we use two generation procedures detailed in algorithms \ref{alg:datagen1} and \ref{alg:datagen2}. Later we refer to them as modes. Factors spanning the subspaces in the first mode are shared between clusters, whereas in the second mode subspaces are independent.  As an input to both procedures we use: $n$ - number of individuals, $SNR$ - signal to noise ratio, $K$ - number of clusters, $p$ - number of variables, $d$ - maximal dimension of a subspace. SNR is the ratio of the power of signal to the power of noise, i.e., $SNR = \frac{\sigma^2}{\sigma^2_e}$ the ratio of variance of the signal to the variance of noise.

\begin{algorithm} [H]
\caption{Data generation with shared factors}
\label{alg:datagen1}          
\begin{algorithmic}                    
    \REQUIRE $n$, $SNR$, $K$, $p$, $d$
    \STATE Number of factors $m \leftarrow K\frac{d}{2}$ 
	\STATE Factors $F = (f_1, \ldots, f_m)$ are generated independently from the multivariate standard normal distribution  and then $F$ is scaled to have columns with mean $0$ and standard deviation $1$ 
	\STATE Draw subspaces dimension $d_1, \ldots d_K$ uniformly from $\{1,\ldots, d \}$  
	\FOR{$i = 1, \ldots, K$}
	\STATE Draw $i$-th subspace basis as sample of size $d_i$ uniformly from columns of $F$ as $F^i$
	\STATE Draw matrix of coefficients $C_i$ from $\mathcal{U}(0.1,1) \cdot sgn(\mathcal{U}(-1,1))$
	\STATE Variables in the $i$-th subspace are $X^i \leftarrow  F^iC_i$
	\ENDFOR
	\STATE Scale matrix $X = (X_1, \ldots, X_K)$ to have columns with unit variance
	\STATE return $X + Z$ where $Z \sim \mathcal{N}(0,\frac{1}{SNR}I_n)$
\end{algorithmic}
\end{algorithm}

\begin{algorithm} [H]
\caption{Data generation with independent subspaces}
\label{alg:datagen2}          
\begin{algorithmic}                    
    \REQUIRE $n$, $SNR$, $K$, $p$, $d$
	\STATE Draw subspaces' dimension $d_1, \ldots d_K$ uniformly from $\{1,\ldots, d \}$  
	\FOR{$i = 1, \ldots, K$}
	\STATE Draw $i$-th subspace basis $F^i$ as sample of size $d_i$ from multivariate standard normal distribution
	\STATE Draw matrix of coefficients $C_i$ from $\mathcal{U}(0.1,1) \cdot sgn(\mathcal{U}(-1,1))$
	\STATE Variables in $i$-th subspace are $X^i \leftarrow  F^iC_i$
	\ENDFOR
	\STATE Scale matrix $X = (X_1, \ldots, X_K)$ to have columns with unit variance
	\STATE return $X + Z$ where $Z \sim \mathcal{N}(0,\frac{1}{SNR}I_n)$
\end{algorithmic}
\end{algorithm}

\subsection{Measures of effectiveness}
To compare clustering produced by our methods we use three measures of effectiveness.
\begin{enumerate}
\item Adjusted Rand Index - one of the most popular measures. Let $A, B$ be the partitions that we compare (one of them should be true partition). Let $a,b,c,d$ denote respectively the number of pairs of points from data set that are in the same cluster both in $A$ and $B$, that are in the same cluster in $A$ but in different clusters in $B$, that are in the same cluster in $B$ but in different clusters in $A$ and that are in the different clusters both in $A$ and $B$. Note that the total number of pairs is $\binom{p}{2}$. Then
\begin{align*}
ARI &= \\
& \frac{\binom{p}{2}(a+d) - [(a+b)(a+c) + (b+d)(c+d)]}{\binom{p}{2}^2- [(a+b)(a+c) + (b+d)(c+d)]}
\end{align*}
The maximum value of ARI is $1$ and when we assume that every clustering is equally probable its expected value is $0$. For details check \cite{ari}.

The next two measures are taken from \cite{soltys}. Let $X = (x_1, \ldots x_p)$ be the data set, $A$ be a partition into clusters $A_1, \ldots A_n$ (true partition) and $B$ be a partition into clusters $B_1, \ldots, B_m$.
\item Integration - for the cluster $A_j$ it is given by formula
\begin{align*}
& Int(A_j) =\\
& \frac{max_{\substack{k = 1, \ldots, m}} \# \{  i \in \{ 1, \ldots p \}: X^i \in A_j \wedge X^i \in B_k \}  }{\# A_j}
\end{align*}

Cluster $B_k$ for which the maximum is reached is called integrating cluster of $A_j$. Integration can be interpreted as the percentage of data points from given cluster of true partition that are in the same cluster in partition $B$. For the whole clustering 
$$
Int(A,B) = \frac{1}{n} \sum_{j=1}^n Int(A_j)
$$

\item Acontamination - for cluster $A_j$ it is given by formula
$$
Acont(A_j) = \frac{ \# \{  i \in \{ 1, \ldots p \}: X^i \in A_j \wedge X^i \in B_k \}  }{\# B_k} 
$$
where $B_k$ is integrating cluster for $A_j$. Idea of acontamination is complementary to integration. It can be interpreted as the percentage of the data in the integrating cluster $B_k$ are from $A_j$. For the whole clustering 
$$
Acont(A,B) = \frac{1}{n} \sum_{j=1}^n Acont(A_j)
$$
\end{enumerate}

Note that the bigger ARI, integration and acontamination are, the better is the clustering. For all three indices the maximal value is $1$.

%\newpage
\subsection{Simulation study results}
In this section we present the outcome of the simulation study. We generate the synthetic data $100$ times. We plot multiple boxplots to compare clusterings of different methods. By default the number of runs (random initializations) is set to $n_{init} = 30$ and the maximal number of iterations within the k-means loop is set to $n_{iter}=30$. Other parameters used in given simulation are written above the plots. They include parameters from data generation algorithms (\ref{alg:datagen1}, \ref{alg:datagen2}) as well as $mode$ indicating which of them was used.

\subsubsection{Generation method}
\begin{figure} 
\centering
\caption{Comparison with respect to the data generation method. Simulation parameters: $n=100, \ p=800, \ K=5, \ d=3, \ SNR=1$.} \label{fig:mode}
\minipage{0.5\textwidth}
  \subcaption{factors not shared}
  \includegraphics[width=\linewidth]{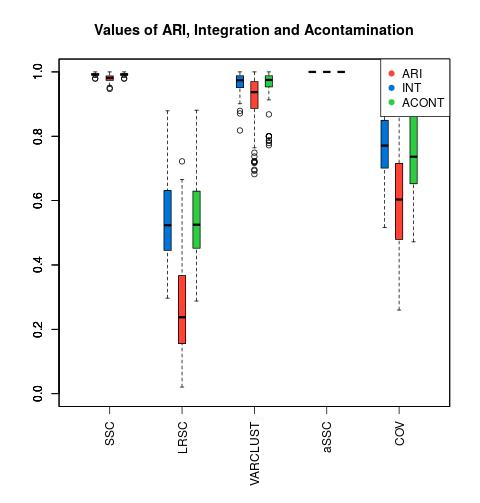}
\endminipage\hfill
\minipage{0.5\textwidth}
  \subcaption{shared factors}
  \includegraphics[width=\linewidth]{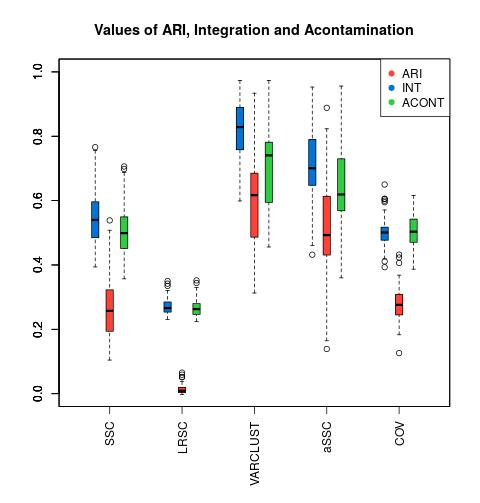}
\endminipage
\end{figure}

In this section we compare the methods with respect to the parameter $mode$, which takes the value $shared$ (data generated using \ref{alg:datagen1}), if the subspaces may share the factors, and the value $not\_shared$ (data generated using \ref{alg:datagen2}) otherwise (Figure \ref{fig:mode}). When the factors are not shared, SSC and VARCLUST provide almost perfect clustering. We can see that in case of shared factors the task is more complex. All the methods give worse results in that case. However, VARCLUST and VARCLUST$_{aSSC}$ outperform all the other procedures and supply acceptable clustering in contrast to SSC, LRSC and COV. The reason for that is the mathematical formulation of SSC and LRSC - they assume that the subspaces are independent and do not have common factors in their bases.
\subsubsection{Number of variables}
\begin{figure*} 
\centering
\caption{Comparison with respect to the number of variables. Simulation parameters: $n=100, \ K=5, \ d=3, \ SNR=1, \ mode : shared$.} 
\label{fig:vars}
\minipage{0.5\textwidth}
	\subcaption{$p=300$}
  \includegraphics[width=\linewidth]{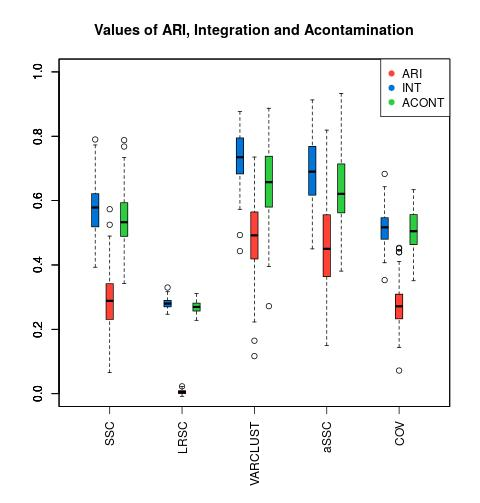}
\endminipage\hfill
\minipage{0.5\textwidth}
\subcaption{$p=600$}
  \includegraphics[width=\linewidth]{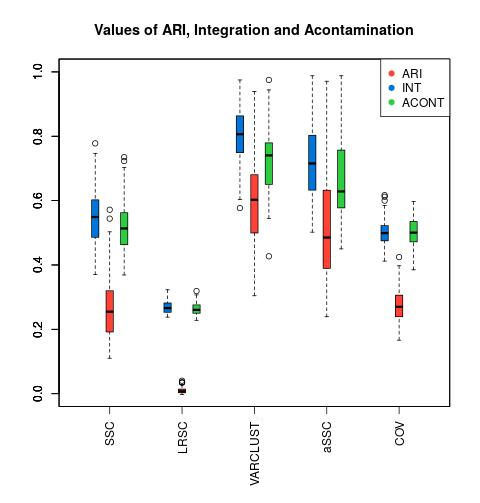}
\endminipage \\
\minipage{0.5\textwidth}
\subcaption{$p=800$}
  \includegraphics[width=\linewidth]{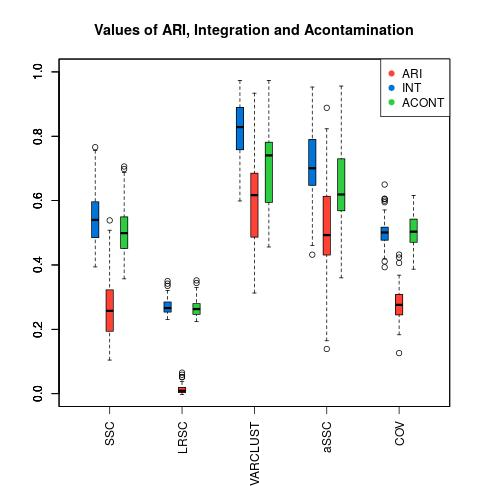}
\endminipage\hfill
\minipage{0.5\textwidth}
\subcaption{$p=1500$}
  \includegraphics[width=\linewidth]{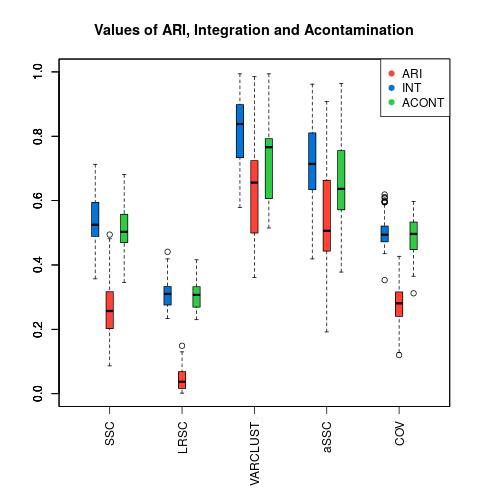}
\endminipage
\end{figure*}

In this section we compare the methods with respect to the number of variables (Figure \ref{fig:vars}). When the number of features increases, VARCLUST tends to produce better clustering. For our method this is an expected effect because when the number of clusters and subspace dimension stay the same we provide more information about the cluster's structure with every additional predictor. Moreover, PESEL from (\ref{bic}) gives a better approximation of the cluster's dimensionality and the task of finding the real model becomes easier. However, for COV, LRSC, SSC this does not hold as the results are nearly identical.
 
\subsubsection{Maximal dimension of subspace}
\begin{figure*} 
\centering
\caption{Comparison with respect to the number of variables. Simulation parameters: $n=100, \ p=600, \ K=5, \ SNR=1, \ mode : shared$. In the left column the maximal dimension passed to VARCLUST was equal to $d$, in the right we passed $2d$.} 
\label{fig:dim}
\minipage{0.45\textwidth}
\subcaption{$d=3$}
  \includegraphics[width=\linewidth]{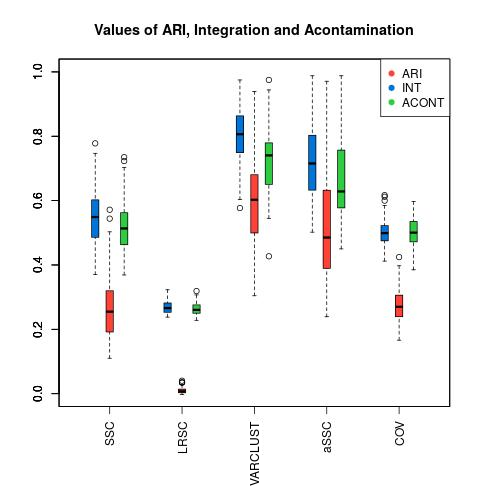}
\endminipage\hfill
\minipage{0.45\textwidth}
\subcaption{$d=3$}
  \includegraphics[width=\linewidth]{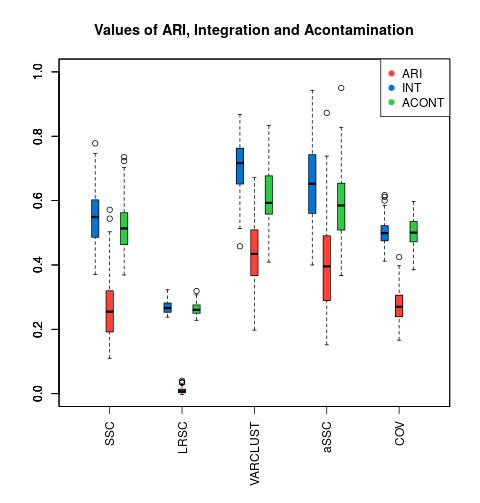}
\endminipage \\
\minipage{0.45\textwidth}
\subcaption{$d=5$}
  \includegraphics[width=\linewidth]{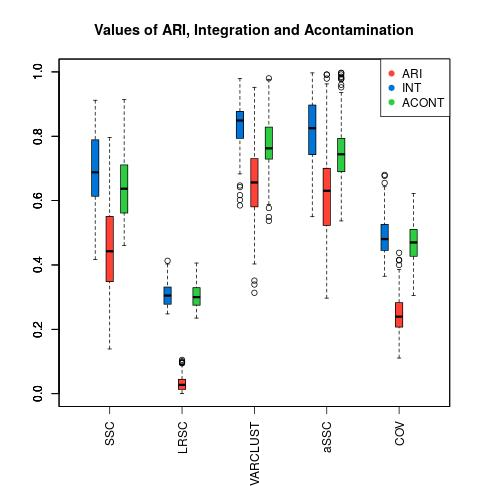}
\endminipage\hfill
\minipage{0.45\textwidth}
\subcaption{$d=5$}
  \includegraphics[width=\linewidth]{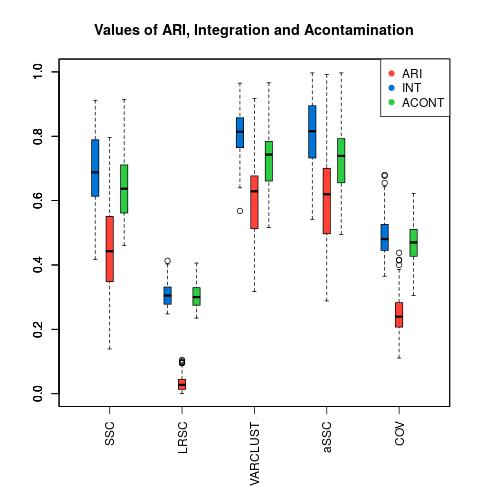}
\endminipage \\
\minipage{0.45\textwidth}
\subcaption{$d=7$}
  \includegraphics[width=\linewidth]{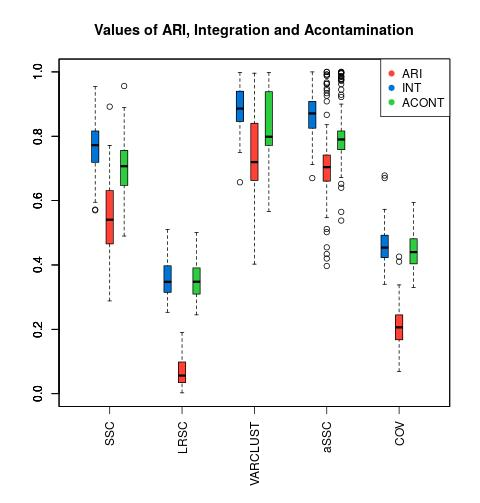}
\endminipage\hfill
\minipage{0.45\textwidth}
\subcaption{$d=7$}
  \includegraphics[width=\linewidth]{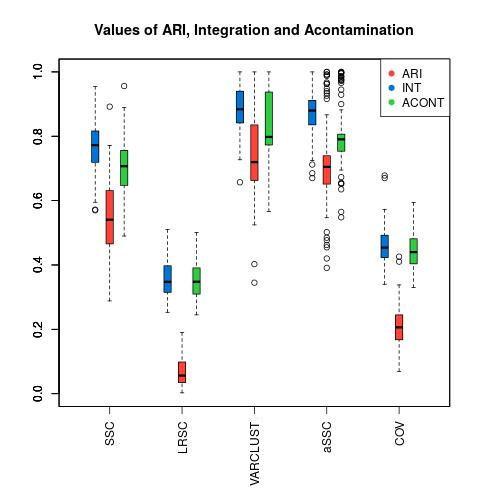}
\endminipage
\end{figure*}

We also check what happens when the number of parameters in the model of VARCLUST increases. In Figure \ref{fig:dim}, in the first column, we compare the methods with respect to the maximal dimension of a subspace ($d = 3,5,7$). However, in real-world clustering problems it is common that it is not known. Therefore, in the second column, we check the performance of VARCLUST and VARCLUST$_{aSSC}$ when the given maximal dimension as a parameter is twice as large as maximal dimension used to generate the data. 

Looking at the first column, we can see that the effectiveness of VARCLUST grows slightly when the maximal dimension increases. However, this effect is not as noticeable as for SSC. It may seem unexpected for VARCLUST but variables from subspaces of higher dimensions are easier to distinguish because their bases have more independent factors. In the second column, the effectiveness of the methods is very similar to the first column except for $d=3$, where the difference is not negligible. Nonetheless, these results indicate that thanks to PESEL, VARCLUST performs well in terms of estimating the dimensions of the subspaces.

\subsubsection{Number of clusters}
\begin{figure*} 
\centering
\caption{Comparison with respect to the number of clusters. Simulation parameters: $n=100, \ p=600, \ d=3, \ SNR=1, \ mode : not \ shared$.} 
\label{fig:clusters}
\minipage{0.5\textwidth}
\subcaption{$K=5$}
  \includegraphics[width=\linewidth]{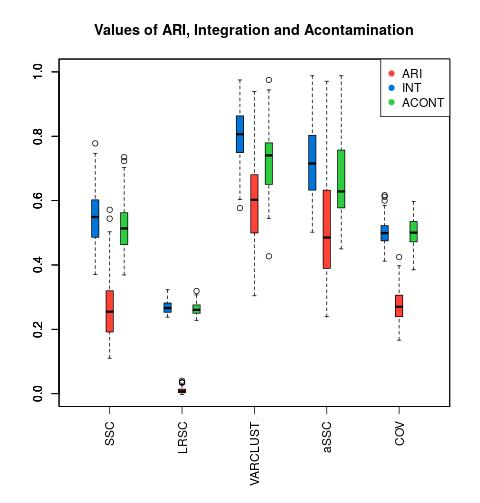}
\endminipage\hfill
\minipage{0.5\textwidth}
\subcaption{$K=10$}
  \includegraphics[width=\linewidth]{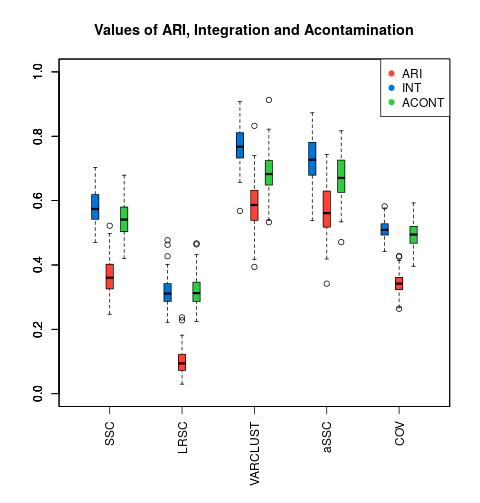}
\endminipage \\
\minipage{0.5\textwidth}
\subcaption{$K=15$}
  \includegraphics[width=\linewidth]{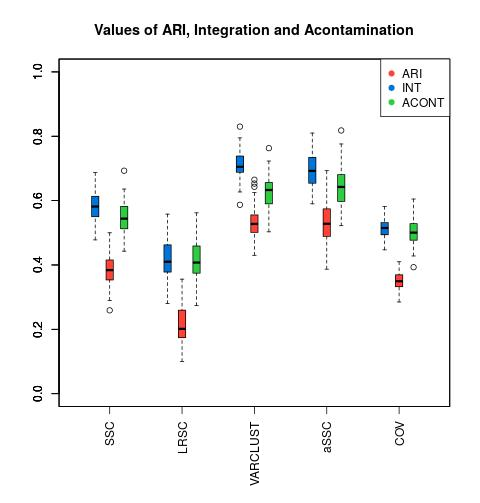}
\endminipage\hfill
\minipage{0.5\textwidth}
\subcaption{$K=20$}
  \includegraphics[width=\linewidth]{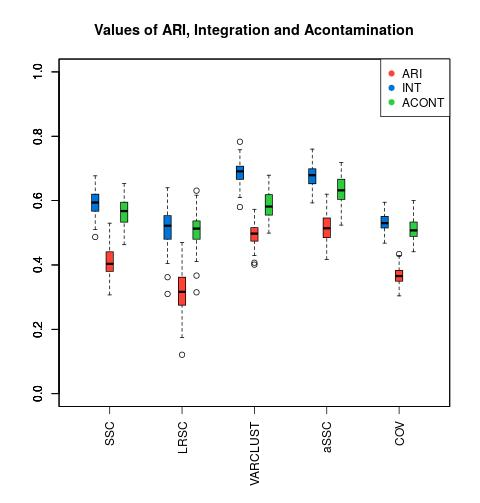}
\endminipage
\end{figure*}

The number of the parameters in the model for VARCLUST grows significantly with the number of clusters in the data set. In Figure \ref{fig:clusters} we can see that for VARCLUST the effectiveness of the clustering diminishes when the number of clusters increases. The reason is the larger number of parameters in our model to estimate. The opposite effect holds for LRSC, SSC and COV, although it is not very apparent. 

\subsubsection{Signal to noise ratio}
\begin{figure*} 
\centering
\caption{Comparison with respect to the signal to noise ratio. Simulation parameters: $n=100, \ p=600, \ K=5, \ d=3, \ mode : not \ shared$.} 
\label{fig:snr}
\minipage{0.5\textwidth}
\subcaption{$SNR=0.5$}
  \includegraphics[width=\linewidth]{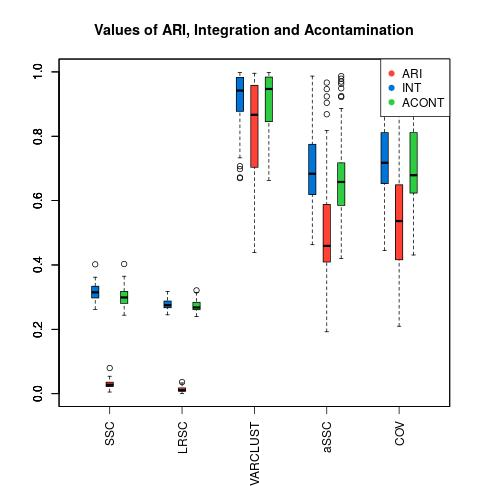}
\endminipage\hfill
\minipage{0.5\textwidth}
\subcaption{$SNR=0.75$}
  \includegraphics[width=\linewidth]{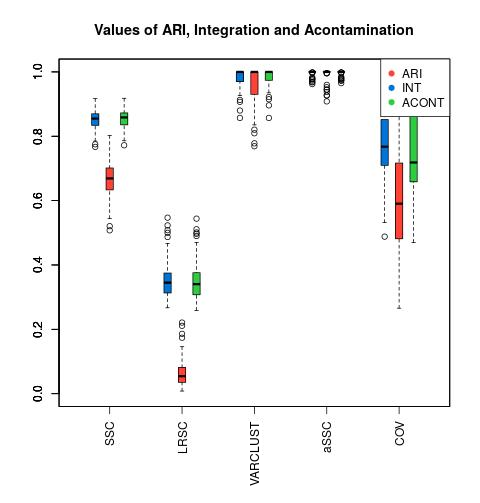}
\endminipage \\
\minipage{0.5\textwidth}
\subcaption{$SNR=1$}
  \includegraphics[width=\linewidth]{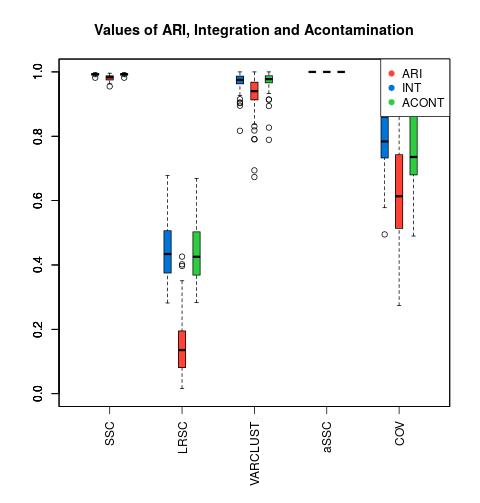}
\endminipage\hfill
\minipage{0.5\textwidth}
\subcaption{$SNR=2$}
  \includegraphics[width=\linewidth]{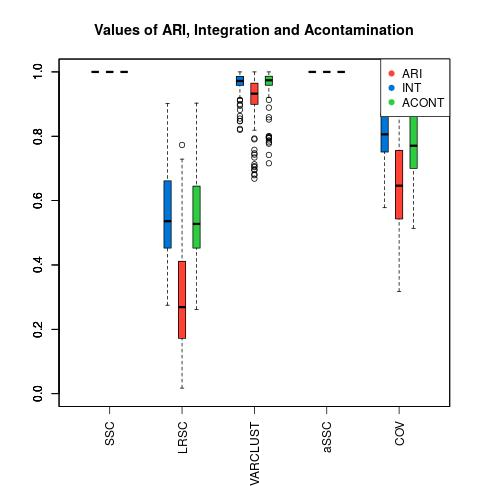}
\endminipage
\end{figure*}

One of the most important characteristics of the data set is signal to noise ratio (SNR). Of course, the problem of clustering is much more difficult when $SNR$ is small because the corruption caused by noise dominates the data. However, it is not uncommon in practice to find data for which $SNR < 1$.

In Figure \ref{fig:snr} we compare our methods with respect to SNR. For $SNR = 0.5$, VARCLUST supplies a decent clustering. In contrary, SSC and LRSC perform poorly. All methods give better results when $SNR$ increases, however for SSC this effect is the most noticeable. For $SNR \geq 1$, SSC produces perfect or almost perfect clustering while VARCLUST performs slightly worse.

\subsubsection{Estimation of the number of clusters}
\begin{figure*} 
\centering
\caption{Estimation of the number of clusters. Simulation parameters: $n=100, \ p=600, \ d=3, \ SNR=1 \ mode : not \ shared$.} 
\label{fig:estim_clust}
\minipage{0.5\textwidth}
\subcaption{$K=5$}
  \includegraphics[width=\linewidth]{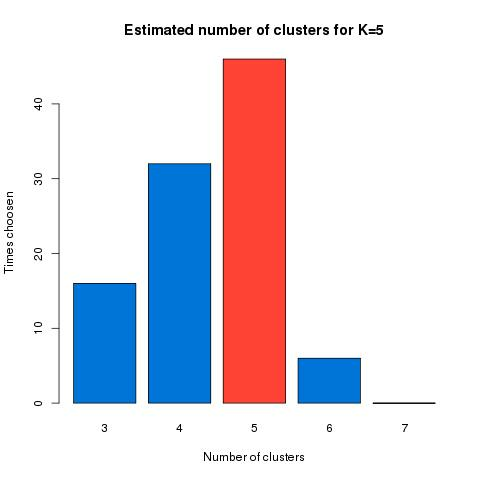}
\endminipage\hfill
\minipage{0.5\textwidth}
\subcaption{$K=10$}
  \includegraphics[width=\linewidth]{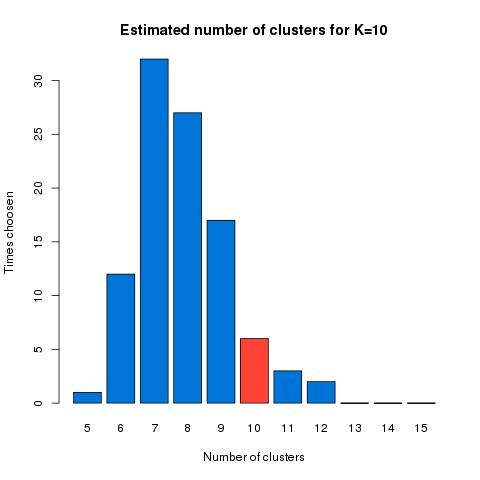}
\endminipage \\
\minipage{0.5\textwidth}
\subcaption{$K=15$}
  \includegraphics[width=\linewidth]{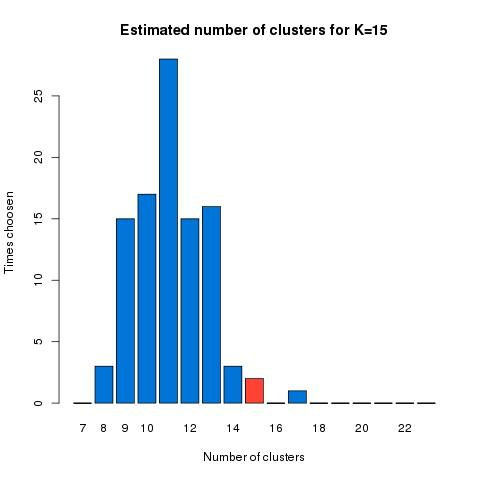}
\endminipage \hfill
\minipage{0.5\textwidth}
\subcaption{$K=20$}
  \includegraphics[width=\linewidth]{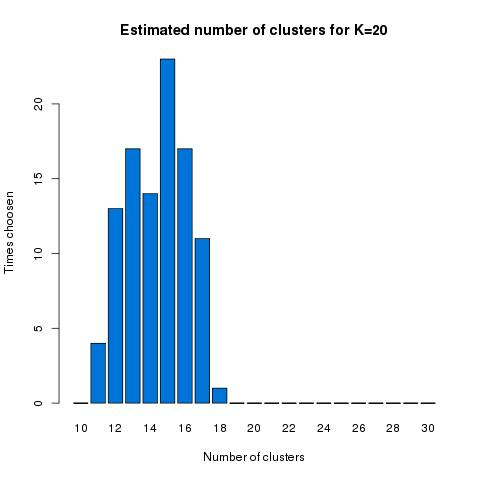}
\endminipage 
\end{figure*}

%Despite a large number of existing methods used in variable clustering problems, there are not many tools for automatic detection of the number of clusters. 
Thanks to mBIC, VARCLUST can be used  for automatic setection of the number of clusters. We generate the data set with given parameters $100$ times and check how often each number of clusters from range $\left[K-\frac{K}{2}, K + \frac{K}{2}\right]$ is chosen (Figure \ref{fig:estim_clust}). We see that for $K=5$ the correct number of clusters was chosen most times. However, when the number of clusters increases, the clustering task becomes more difficult, the number of parameters in the model grows and VARCLUST tends to underestimate the number of clusters.

\subsubsection{Number of iterations}
\begin{figure*} \centering
\caption{mBIC with respect to the number of iterations for 4 different initializations. Simulation parameters: $n=100, \ K=5, \ d=3, \ SNR=1 \ mode : shared$.} 
\label{fig:iter}
\minipage{0.5\textwidth}
\subcaption{$p=750$}
  \includegraphics[width=\linewidth]{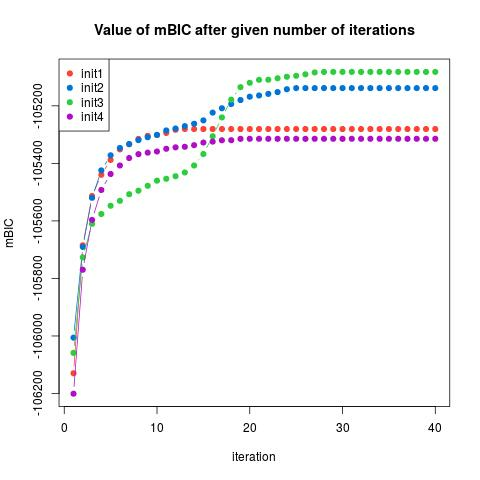}
\endminipage\hfill
\minipage{0.5\textwidth}
\subcaption{$p=1500$}
  \includegraphics[width=\linewidth]{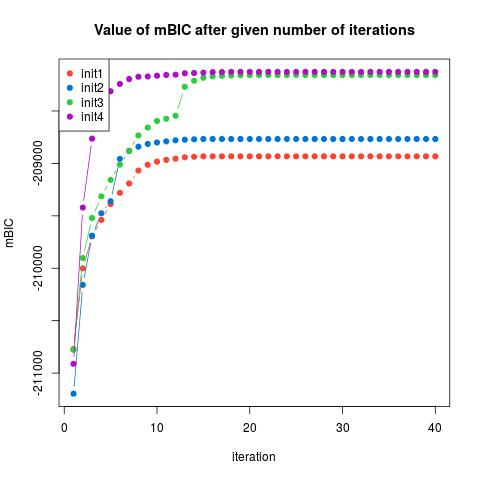}
\endminipage \\
\minipage{0.5\textwidth}
\subcaption{$p=3000$}
  \includegraphics[width=\linewidth]{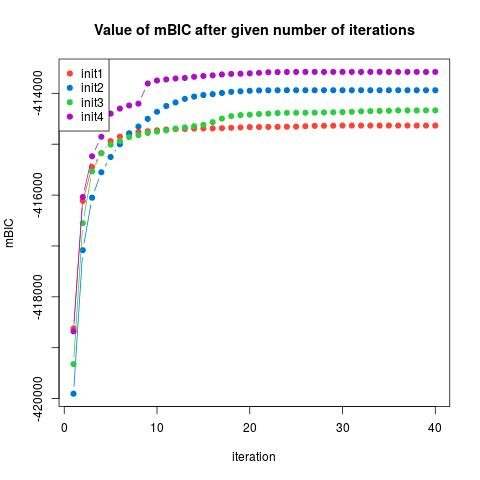}
\endminipage \hfill
\minipage{0.5\textwidth}
\endminipage 
\end{figure*}

In this section we investigate convergence of mBIC within $k$-means loop for four different initializations (Figure \ref{fig:iter}). We can see that it is quite fast: in most cases it needed no more than 20 iterations of the $k$-means loop. We can also notice that the size of the data set (in this case the number of variables) has only small impact on the number of iterations needed till convergence. However, the results in Figure \ref{fig:iter} show that multiple random initializations in our algorithm are required to get satisfying results - the value of mBIC criterion varies a lot between different initializations. 

\subsubsection{Execution time}
\begin{figure} \centering
\caption{Comparison of the execution time of the methods with respect to $p$ and $K$. Simulation parameters:$n=100, \ d=3, \ SNR=1 \ mode : shared$.} 
\label{fig:time}
\minipage{0.5\textwidth}
\subcaption{With respect to the number of variables}
  \includegraphics[width=\linewidth]{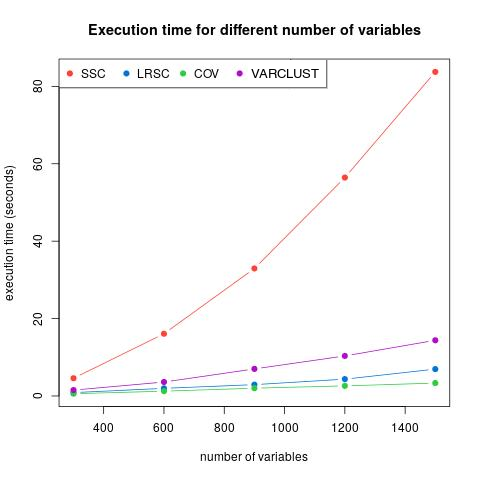}
\endminipage\hfill
\minipage{0.5\textwidth}
\subcaption{With respect to the number of clusters}
  \includegraphics[width=\linewidth]{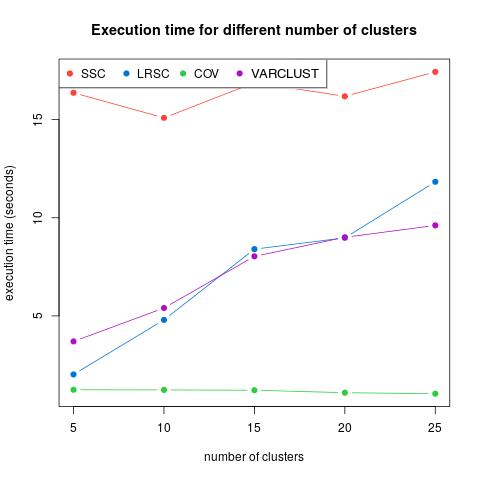}
\endminipage
\end{figure}

In this section we compare the execution times of compared methods. They were obtained on the machine with \textit{Intel(R) Core(TM) i7-4790 CPU 3.60GHz, 8 GB RAM}. The results are in Figure \ref{fig:time}. For the left plot $K=5$ and for the right one $p=600$. On the plots for both VARCLUST and COV we used only one random initialization. Therefore, we note that for $n_{init}=30$ the execution time of VARCLUST will be larger. However, not by exact factor of $n_{init}$ thanks to parallel implementation in \cite{sobczyk_varclust_package}. Nonetheless, VARCLUST is the most computationally complex of these methods. We can see that COV and SSC do not take longer for bigger number of clusters when the opposite holds for VARCLUST and LRSC. What is more, when the number of variables increases, the execution time of SSC grows much more rapidly than time of one run of VARCLUST. Therefore, for bigger data sets it is possible to test more random initializations of VARCLUST in the same time as computation of SSC. Furthermore, running VARCLUST with segmentation returned by SSC (enhancing the clustering) is not much more time consuming than SSC itself.

\subsubsection{Discussion of the results}
The simulation results prove that VARCLUST is an appropriate method for variable clustering. As one of the very few approaches, it is adapted to the data dominated by noise. One of its biggest advantages is a possibility to recognize subspaces which share factors. It is also quite robust to increase in the maximal dimension of a subspace. Furthermore, it can be used to detect the number of clusters in the data set. Last but not least, in every setting of the parameters used in our simulation, VARCLUST outperformed LRSC and COV and did better or as well as SSC. The main disadvantage of VARCLUST is its computational complexity. Therefore, to reduce the execution time one can provide custom initialization as in VARCLUST$_{aSSC}$. This method in all cases provided better results than SSC, so our algorithm can also be used to enhance the clustering results of the other methods. The other disadvantage of VARCLUST is a problem with the choice of the parameters $n_{init}$ or $n_{iter}$. Unfortunately, when data size increases, in order to get acceptable clustering we have to increase at least one of these two values. However, it is worth mentioning that in case of parameters used in out tests $n_{init} = 30$ and the maximal number of iterations equal to $30$ on a machine with $8$ cores the execution time of VARCLUST is comparable with execution time of SSC.

\section{Applications to  real data analysis} \label{sec:realdata}
In this section we apply  VARCLUST to two different data sets and show that our algorithm can produce meaningful, interpretable clustering and dimensionality reduction.
\subsection{Meteorological data}
%\mb{If the dimensions are large enough - we expect to see convergence, illustrate it empirically.}\\%2.  (studied by M. Staniak. Where is the file? get the results from M. Staniak)
%\section{Case study}
%\end{document}
First, we will analyze air pollution data from Krak{\'o}w, Poland \cite{airly}.
This example will also serve as a short introduction to the \texttt{varclust} R package.
\subsubsection{About the data}
Krakow is one of the most polluted cities in Poland and even in the world.
This issue has gained enough recognition to inspire several grass-root initiatives that aim to monitor air quality and inform citizens about health risks.
\textit{Airly} project created a huge network of air quality sensors which were deployed across the city.
Information gathered by the network is accessible via the \url{map.airly.eu} website.
Each of 56 sensors measures temperature, pressure, humidity and levels of particulate matters PM1, PM2.5 and PM10 (number corresponds to the mean diameter).
This way, air quality is described by 336 variables.
Measurements are done on an hourly basis.

Here, we used data from one month.
We chose {M}arch, because in this month the number of missing values is the smallest.
First, we remove{d} non-numerical variables from the data set.
We remove columns with a high percentage (over 50\%) of missing values and impute the other by the mean. 
We {used} two versions of the data set: \texttt{march\_less} data frame {containing} hourly measurements (in this case number of observations is greater than number of variables) and \texttt{march\_daily} containing averaged daily measurements (which satisfies the $p \gg n$ assumption). 
Results for both versions are consistent.
The dimensions of the data are $577{\times}263$ and $25{\times}263$, respectively.
Both data sets along with R code and results are available on {\url{https://github.com/mstaniak/varclust_example}}

\subsubsection{Clustering based on random initialization}

When the number of clusters is not known, we can use the {\texttt{mlcc.bic}} function which finds a clustering of variables with an estimated  number of clusters and also returns factors that span each cluster.
A minimal call to {\texttt{mlcc.bic}} function requires just the name of a data frame in which the data are stored.

{
%{\tiny
\begin{lstlisting}
varclust_minimal <-
mlcc.bic(march_less, greedy = F)
\end{lstlisting}
}
%}

The returned object is a list containing the resulting segmentation of variables (\texttt{segmentation} element), a list with matrices of factors for each cluster, mBIC for the chosen model, list describing dimensionality of each cluster and models fitted in other iterations of the algorithm (non-optimal models). 
By default, at most 30 iterations of the greedy algorithm are used to pick a model.
Also by default it is assumed that the number of clusters is between 1 and 10, and the maximum dimension of a single cluster is 4. 
These parameters can be tweaked.
Based on comparison of mBIC values for clustering results with different maximum dimensions, we selected 6 as the maximum dimension.

{
%{\tiny{
\begin{lstlisting}
varclust_clusters = 
mlcc.bic(march_less,greedy = TRUE,
flat.prior = TRUE, max.dim = 6)
\end{lstlisting}
}
%}

To minimize the impact of random initialization, we can run the algorithm many times and select best clustering based on the value of mBIC criterion.
We present results for one of clusterings obtained this way.

% In our example, five clusters were chosen, so the the value is far from both extreme cases 1 and 10.
% On the other hand, estimated number of dimension is 4 in all cases, so we may check if a larger maximum dimension can provide a better fit.
% \begin{lstlisting}
% daily_varclust2 <- mlcc.bic(march_daily,
%                             greedy = FALSE,
%                             max.dim = 8)
% \end{lstlisting}
% Value of BIC of the final model is smaller than the BIC value for the model with maximum dimension set to 4, so the increase wasn't necessary.
% The estimated number of clusters is 6.
% Variables describing different types of measurement temperature, pressure and humidity were almost perfectly separated into three clusters. 
% Turns out that each group of variables is spanned by 4 factors.
We can see that variables describing temperature, humidity and pressure were grouped in four clusters (with pressure divided into two clusters and homogenous clusters for humidity and temperature related variables), while variables that describe levels of particulate matters are spread among different clusters that do not describe simply one size of particulate matter (1, 2.5 or 10), which may imply that measurements are in a sense non-homogenous. 
In Figure \ref{ex:map} we show how these clusters are related to geographical locations.

\subsubsection{Clustering based on SSC algorithm}

The {\texttt{mlcc.bic}} function performs clustering based on a random initial segmentation.
When the number of clusters is known or can be safely assumed, we can use the \texttt{mlcc.reps} function, which can start from given initial segmentations or a random segmentation.
We will show how to initialize the clustering algorithm with a fixed grouping. 
For illustration, we will use results of Sparse Subspace Clustering (SSC) algorithm.
SSC is implemented in a Matlab package maintained by Ehsan Elhamifar \cite{ssc}.
As of now, no R implementation of SSC is available.
We store resulting segmentations for numbers of clusters from 1 to 20 in vectors called 
\texttt{clx}, where x is the number of clusters.
Now the calls to \texttt{mlcc.reps} function should look like the following example. 

{{
{\small{
\begin{lstlisting}
vclust10 <- mlcc.reps(march_less,
	numb.clusters=10,max.iter=50,
 	initial.segmentations=list(cl10))
\end{lstlisting}
}}

The result is a list with a number of clusters (\texttt{segmentation}), calculated mBIC and a list of factors spanning each of the clusters.
For both initialization methods, variability of results regarding the number of clusters diminished by increasing the \texttt{numb.runs} argument to \texttt{mlcc.bic} and \texttt{mlcc.reps} functions which control the number of runs of the k-means algorithm.
% Checking some larger values of number of clusters reveals that the criterion is stable above around 8 clusters, so we choose it as the optimal number of clusters.
% 
% \begin{figure}[htbp] \centering
% \label{ex:mbic}
% \centering
% \includegraphics[width=7cm]{nclust_vs_BIC.png}
% \caption{Value of mBIC criterion for fixed number of clusters and initial segmentation generated by SSC algorithm.}
% \end{figure}
% 
% Variables describing temperature, humidity and pressure were grouped together in 4 clusters and variables related to particular matter levels were spread across remaining 4 clusters.
% For each station, all 3 variables describing PM1, PM10 and PM2.5 fall into the same cluster.
% In Figure \ref{ex:map} we show how these clusters are related to geographical location of a station.

\begin{figure}[htbp] \centering
\centering
\includegraphics[width=10cm]{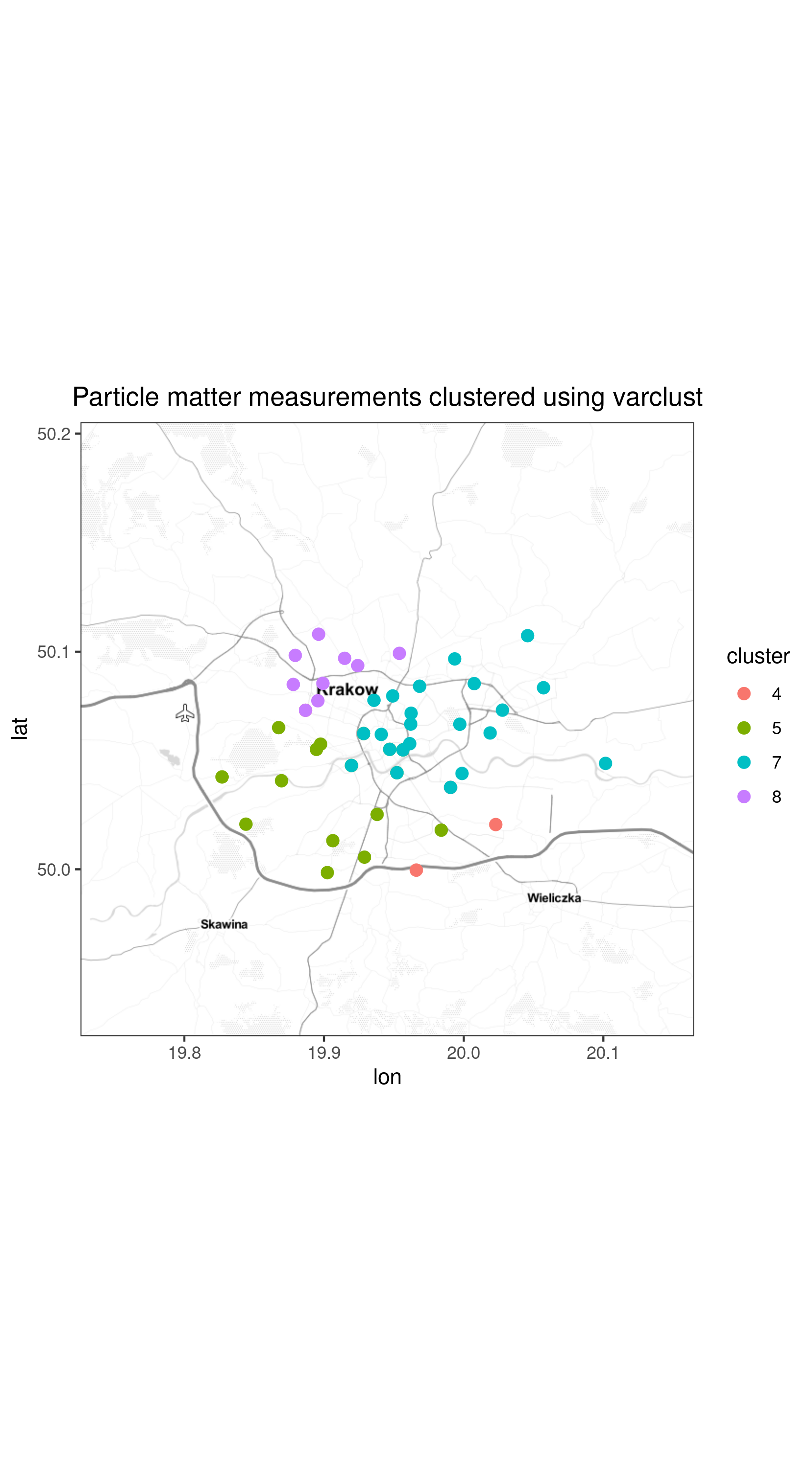}
\caption{Clusters of variables describing particulate matter levels on a map of Krakow. {Without any prior knowledge on spatial structure, VARCLUST groups variables corresponding to sensors located near each other.}}
\label{ex:map}
\end{figure}

\subsubsection{Conclusions}

We applied VARCLUST algorithm to data describing air quality in Krak{\'o}w.
We were able to reduce the dimensionality of the data significantly. 
It turns out that for each  characteristics: temperature, humidity and
the pressure, measurements made in 56 locations can be well represented by
a low dimensional projection found by Varclust.
Additionally, variables describing different particulate matter levels can be clustered into geographically meaningful groups, clearly separating the center and a few bordering regions.
If we were to use these measurements as explanatory variables in  a model describing for example effects of air pollution on health, factors that span clusters could be used instead as predictors, allowing for a significant dimension reduction.

The results of the clustering are random by default.
Increasing the number of runs of $k$-means algorithm and maximum number of iterations of the algorithm stabilize the results.
Increasing these parameters also increases the computation time.
Another way to remove randomness is to select an initial clustering using another method.
In the examples, clustering based on SSC algorithm was used.

The \texttt{mlcc.bic} function performs greedy search by default, meaning that the search stops after first decrease in mBIC score occurs.
On the one hand, this might lead to suboptimal choice of number of clusters, so setting \texttt{greedy} argument to \texttt{FALSE} might be helpful, but on the other hand, the criterion may become unstable for some larger numbers of clusters.

\subsection{TCGA Breast Cancer Data}

In the next subsection, the VARCLUST clustering method is applied on large open-source data generated by The Cancer Genome Atlas (TCGA) Research Network, available on \url{http://cancergenome.nih.gov/}.
TCGA has profiled and analyzed large numbers of human tumours to discover molecular aberrations at the DNA, RNA, protein, and epigenetic levels. 
In this analysis, we focus on the Breast Cancer cohort, made up of all patients reviewed by the TCGA Research Network, including all stages and all anatomopathological characteristics of the primary breast cancer disease, as in \cite{TCGAnature}. 

The genetic informations in tumoral tissues DNA that are involved in gene expression are measured from messenger
%The genetic informationÂs in tumoral tissues DNA that are involved in gene expression are measured from messenger 
RNA (mRNA) sequencing. The analysed data set is composed of $p=60488$ mRNA transcripts for $n=1208$ patients.

%\fp{Valerie, can you give a detailed description including origin of the patients for instance, date of the analysis (at the beginning of the cancer), on which type of patients (Triple Neg,\ldots) I think that such descriptions are important in such a paper.}
 % c'est fait

For this data set, our objective is twofold. First, from a machine learning point of view, we hope that this clustering procedure will provide a sufficiently efficient dimension reduction in order to improve the forecasting issues related to the cancer,  for instance  the prediction of the reaction of patients to a given treatment or  the life expectancy in terms of the transcriptomic diagnostic. \\
\noindent Second, from a biological point of view, the clusters of gene expression might be interpreted as distinct biological processes. Then, a way of measuring the quality of the VARCLUST  method is to compare the composition of the selected clusters with some biological pathways classification (see Figure \ref{fig:biological1}). More precisely, the goal is to check if the clusters constructed by VARCLUST correspond to already known biological pathways (Gene Ontology, \cite{geneontology}).

\subsubsection{Data extraction and gene annotations} 
 This ontological classification aims at doing a census of all described biological pathways. To grasp the subtleties inherent to biology, it is important to keep in mind that one gene may be involved in several biological pathways and that most of biological pathways are slot or associated with each other. The number of terms on per Biological process ontology was $29687$ in January 2019 while the number of protein coding genes is around $20000$. Therefore, one cannot consider each identified biological process  as independent characteristic.
 
 The RNASeq raw counts were extracted from the TCGA data portal. The scaling normalization and log transformation (\cite{robinson2010}) were computed using voom function (\cite{voom}) from limma package  version 3.38.3
(\cite{ritchie2015}). The gene annotation was realised with biomaRt package  version 2.38.0  (\cite{durink2005}, \cite{durink2009}).

The enrichment process aims to retrieve a functional profile of a given set of genes in order to better understand the underlying biological processes. Therefore, we compare the input gene set (i.e, the genes in each cluster) to each of the terms in the gene ontology. A statistical test can be performed for each bin to see if it is enriched for the input genes. It should {be mentioned} that all genes in the input genes may not be {retrieved} in the Gene Ontology Biological {P}rocess and conversely, all genes in the Biological Process may not be present in the input gene set.
To perform the GO enrichment analysis, we used GoFuncR package \cite{GOfuncR} version 1.2.0. Only Biological Processes identified with Family-wise Error Rate $\text{p-value} < 0.05$ were reviewed. Data processing and annotation enrichment were performed using \textsc{R} software version 3.5.2.

\begin{figure}[htbp] \centering
      \includegraphics[width=0.5\textwidth]{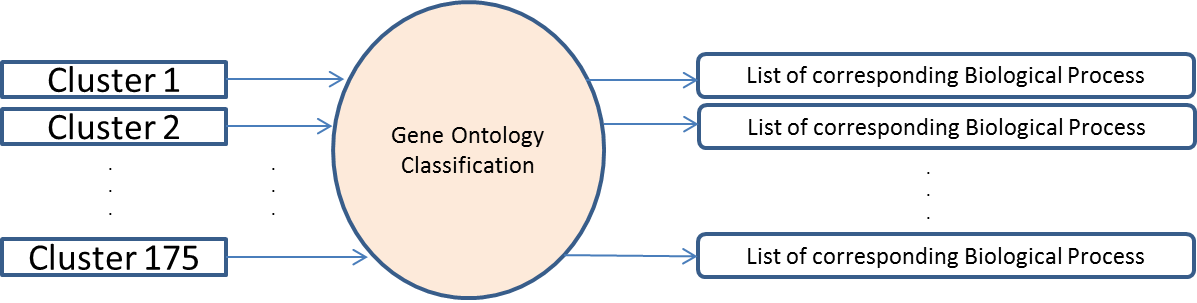}
      \caption{Bioinformatic annotation process for each cluster identified by VARCLUST}
      \label{fig:biological1}
\end{figure}

\subsubsection{Evolution of the mBIC and clusters strucure}
The number of clusters to test was fixed to $50$, $100$, $150$, $175$, $200$, $225$, $250$. The maximal subspace dimension was fixed to 8, the number of runs was 40, and the maximal number of iterations of the algorithm was 30.

As illustrated in the Figure 12, the mBICs remain stable from the 35th iteration.
% Fabien, je ne sais pas comment indexer le numÃ©ro de la figure automatiquement
% je ne sais pas comment interprÃ©ter le fait que le mBIC augmente doucement, mais continue Ã  augmenter avec le nombre d'itÃ©rations.
The mBIC is not $a.s.$ increasing between $50$ and $250$ clusters sets. The mBIC for $K=175$ and $K=250$ clusters sets were close. 
The proportion of clusters with only one principal component is also higher for $K=175$ and $K=250$ clusters sets.  
 
\begin{figure*} \centering
\begin{tabular}{c c c}
\includegraphics[width=0.33\textwidth]{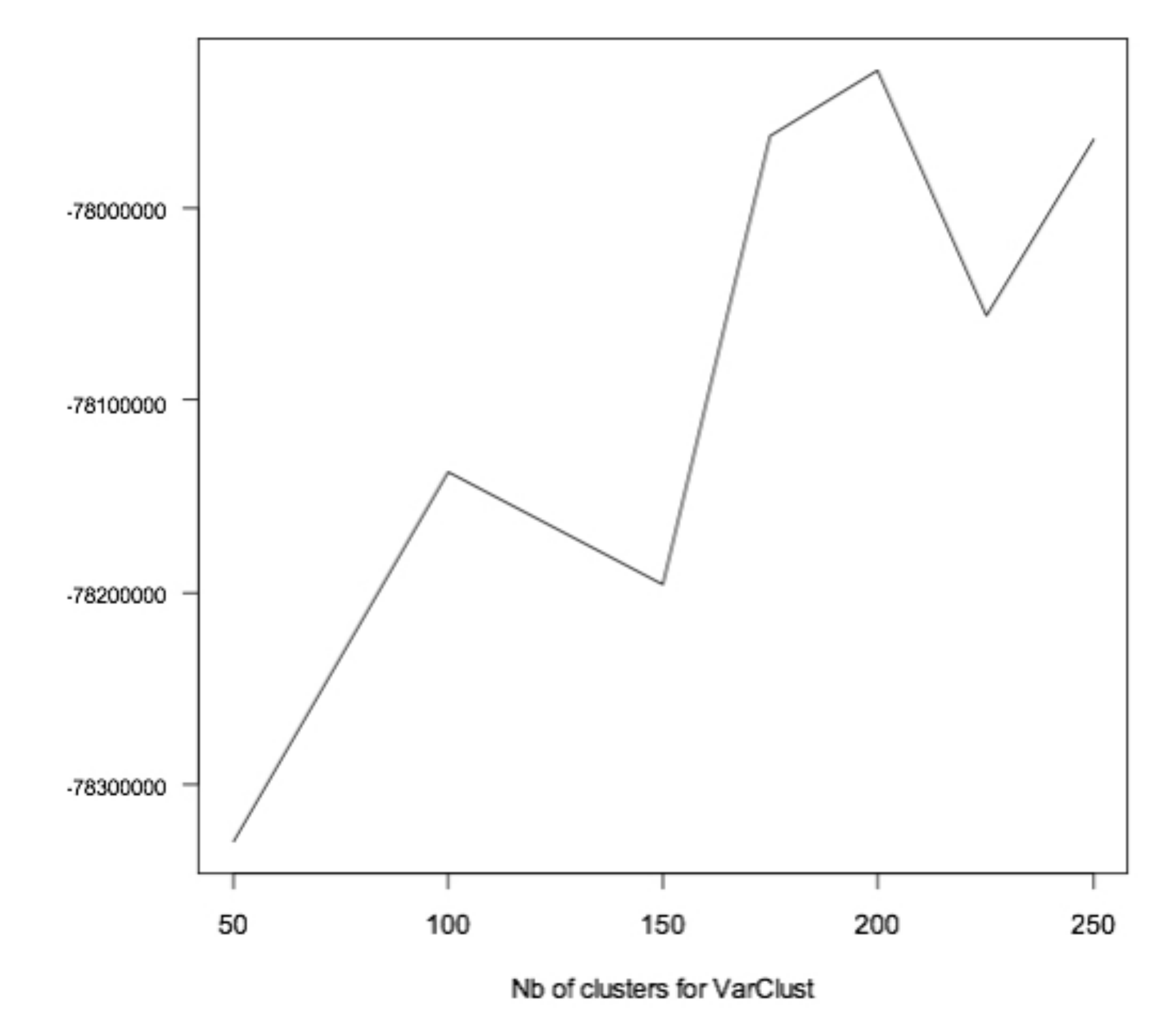} &
\includegraphics[width=0.33\textwidth]{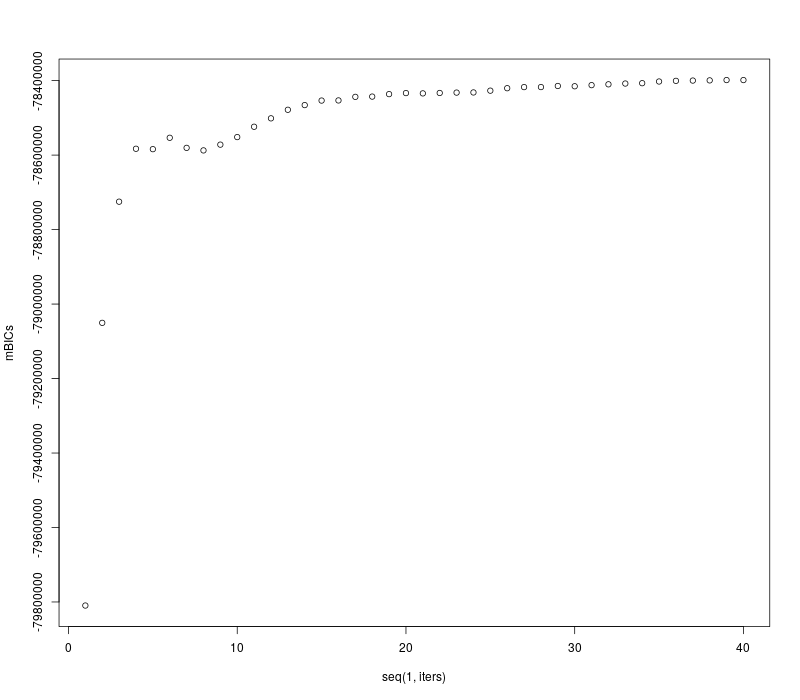} &
\includegraphics[width=0.33\textwidth]{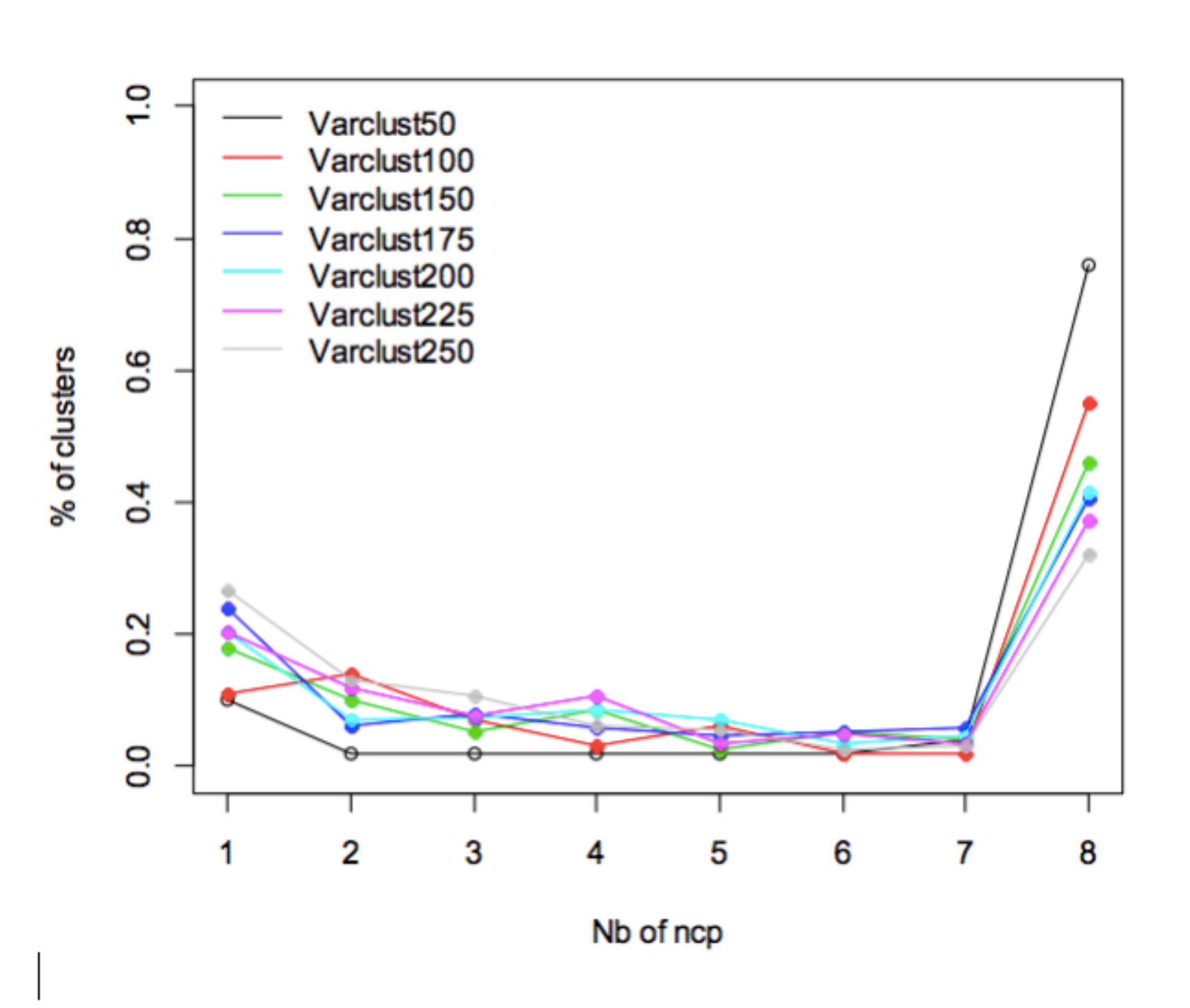}
\end{tabular}

 \caption{Left: evolution of the mBIC with the number of clusters; middle: evolution of the mBIC with the number of iterations; right: number of principal components in clusters  in terms of $K$.}
\label{mBICTCGA}
\end{figure*} 
 
%\begin{figure} \centering
%   \begin{minipage}[c]{.32\linewidth}
%      \includegraphics{mBIC_TCGA.pdf}
%   \end{minipage} \hfill
%   \begin{minipage}[c]{.32\linewidth}
%    \includegraphics{mbic175_40.png}
% \end{minipage}
%  \begin{minipage}[c]{.32\linewidth}
%    \includegraphics{ncp_TCGA.pdf}
% \end{minipage}
% \caption{Left: evolution of the mBIC with the number of clusters; middle: evolution of the mBIC with the number of iterations; right: number of principal components in clusters  in terms of $K$.}
%\label{mBICTCGA}
%\end{figure}

%\begin{figure}[htbp] \centering
   %\begin{minipage}[c]{.46\linewidth}
      %\includegraphics{numbergene.pdf}
   %\end{minipage} \hfill
  
% \caption{Left: Number of genes per cluster, $K=200$. Right: Number of principal components in clusters  in terms of $K$}
 %\label{numbergenecluster}
%\end{figure}

% Fabien, est-il possible de joindre la figure de droite aux 2 figures prÃ©cÃ©dentes?
% La figure de gauche n'est plus d'actualitÃ©

\subsubsection{Biological specificity of clusters}
%We hypothesized that clusters with few principal component might correspond to a potential delimitated biological entity. %Due to the complexity of sloted Biological Process, we focus on clusters with only one Biological Process ``GO-identified''. 
In this subsection, we focus on some biological interpretations in the case:  $K=175$ clusters.  
%(where the mBIC and the proportion of unidimensional clusters were high).\\

\noindent In order to illustrate the correspondance between the genes clustering and the
biological annotations in Gene Ontology, we have selected one cluster
with only one Gene Ontology Biological Process (Cluster number $3$) and
one cluster with two Gene Ontology Biological processes (Cluster
 number $88$). We keep this numbering notation in the sequel.
 
 %\begin{figure}[htbp] \centering
  % \begin{minipage}[c]{\linewidth}
%\includegraphics[width=5cm]{clusters_12.pdf}
%\end{minipage}
%\caption{List of clusters with one or two GO Biological processes.}
%\%label{clusterone}
%\end{figure}

\noindent Among the $98$ genes in Cluster $3$,  $70$ ($71.4$\%, called
``Specific Genes'') were reported in the  GO Biological process named
\textit{calcium-independent cell-cell adhesion via plasma membrane,
cell-adhesion molecules} ($GO:0016338$). The number of principal
components in this cluster was $8$ (which may indicate that one
Biological process has to be modeled using many components). Among the
$441$ genes in Cluster $88$,  $288$ ($65.3$\%) were reported in the
GO Biological processes named \textit{small molecule metabolic process} (
$GO:0044281$) and \textit{cell-substrate adhesion} ($GO:0031589$).  The number of
principal components in this cluster was also $8$.

\noindent To investigate whether the \textit{specific} genes, $i.e.$ involved in the GO biological process
 are well separated from \textit{unspecific} genes  (not involved in the GO biological process), we computed two standard PCAs in Clusters $3$ and $88$ separetely. As shown in Figure \ref{fig:PCAcluster}, the
separation is well done. 
%Three clusters had only one GO Biological Process identified (numbers $3$, $79$ and $110$). The pathways were calcium-independent cell-cell adhesion via plasma membrane, cell-adhesion molecules (GO:0016338), ``protein folding'' (GO:GO:0006457) and ``cornification'' (GO:0070268). In these clusters, the number of transcript genes identified in the corresponding GO Biological processes were 70/98 (71.4 \%),  10/251 (4.0 \%) and 10/21 (47.6 \%).
%
%The number of principal components in these clusters were 8, 8 and 6 respectively. This may indicate that one Biological process has to be modeled using many components.To investigate whether the genes involved in the GO biological process are well separated from genes present in the cluster but not in the GO biologcal process, we compute a standard PCA. In the cluster $3$, the separation is well done for the 2 principal components.

\begin{figure*}[htbp] \centering
\begin{tabular}{c c c}
	\includegraphics[width=0.49\textwidth]{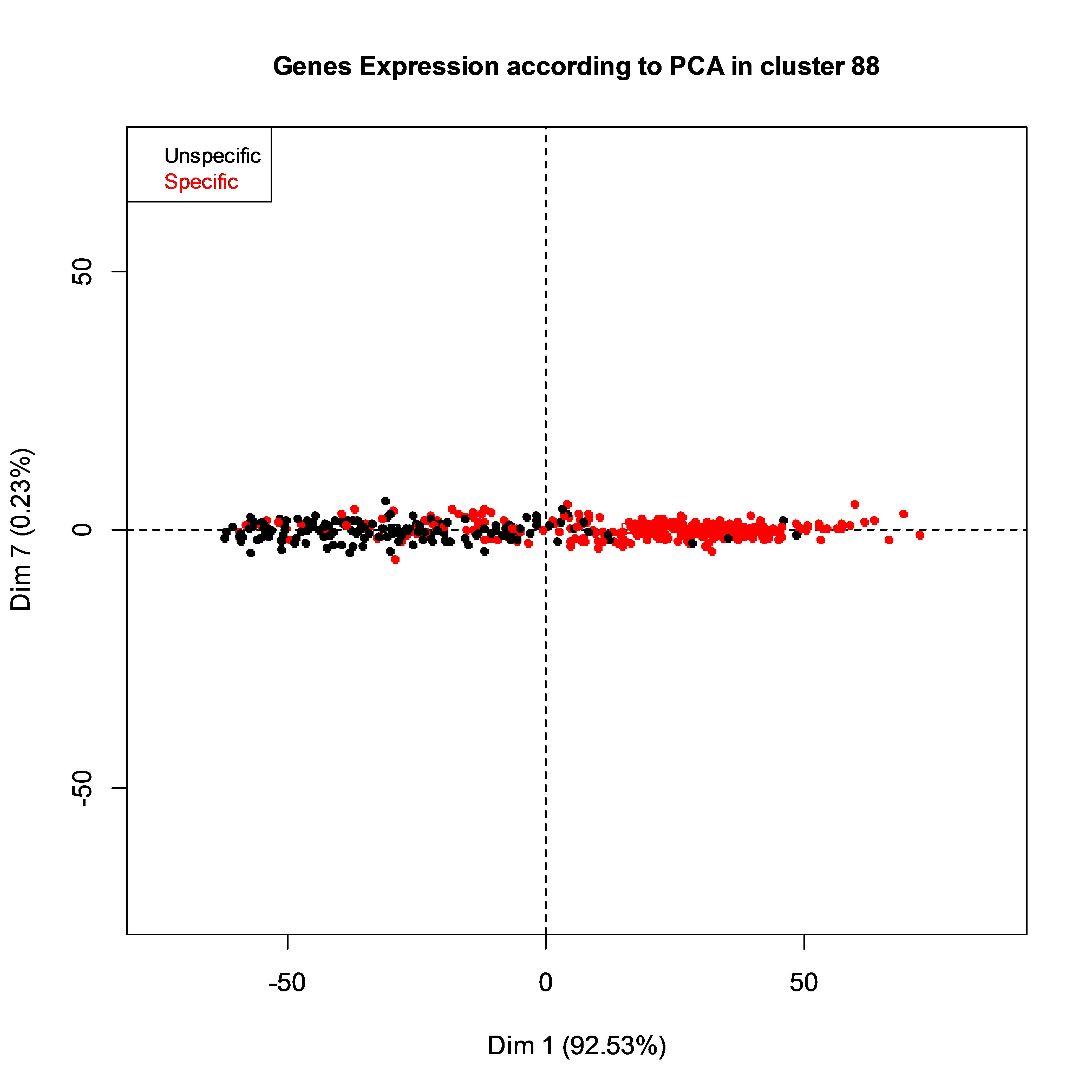} &
	\includegraphics[width=0.49\textwidth]{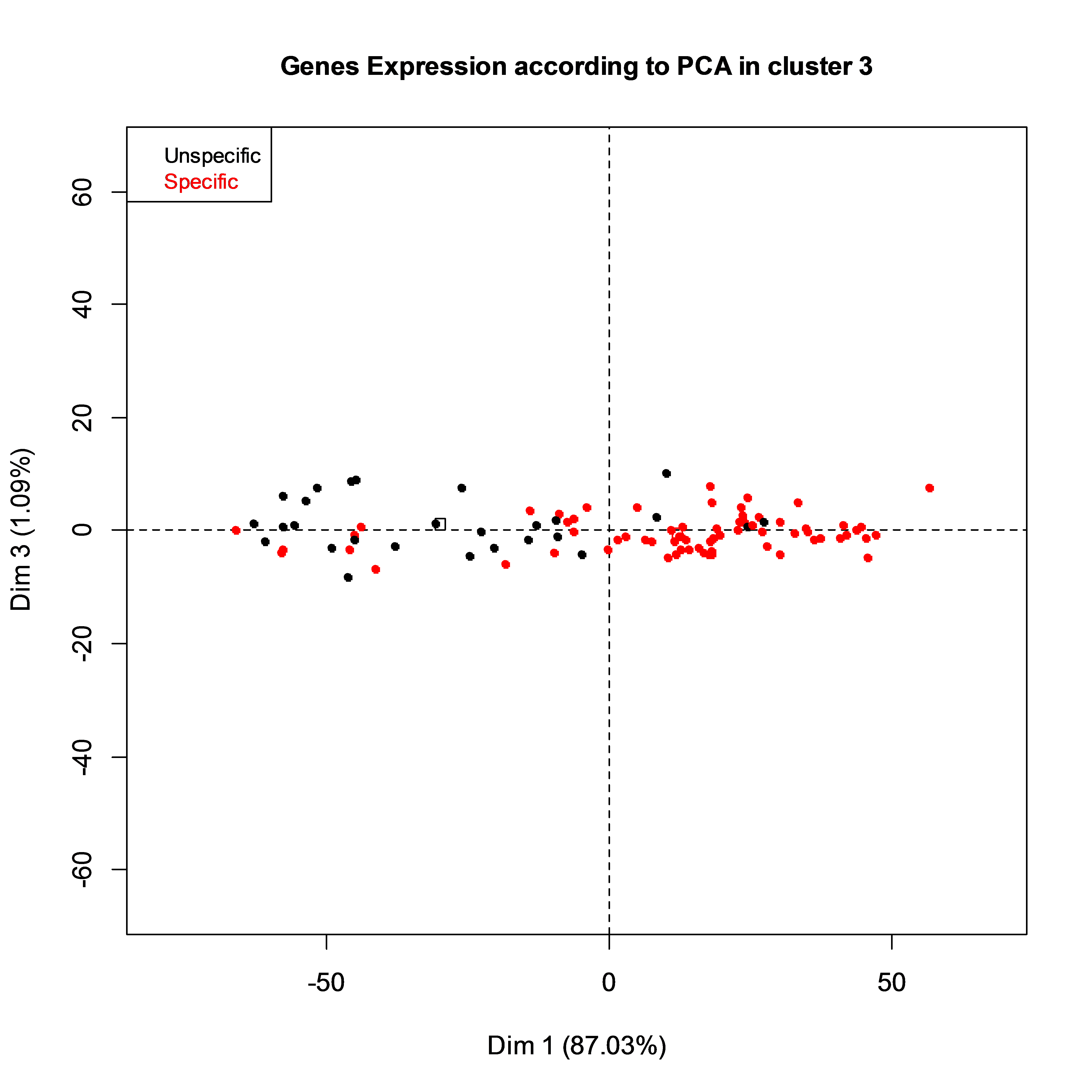} 
\end{tabular}
 \caption{Repartition of specific (red color)  and unspecific genes (black color) according to a standard PCA.}
\label{fig:PCAcluster}
\end{figure*}

\section{VARCLUST package}\label{sec:package}

The package \cite{sobczyk_varclust_package} is an R package that implements VARCLUST algorithm. To install it, run \verb| install.packages("varclust") | in R console. The main function is called \emph{mlcc.bic} and it provides estimation of:
\begin{itemize}
\item Number of clusters $K$
\item Clusters dimensions $\vec{k}$
\item Variables segmentation  $\Pi$
\end{itemize}
These estimators minimize modified BIC described in Section \ref{sec:math-model-varclust}.
\\
For the whole documentation use \verb|?mlcc.bic|. Apart from running VARCLUST algorithm using random initializations, the package allows for a hot start specified by the user.

%\subsection{Choosing values of parameters}

Information about all parameters can be found in the package documentation. Let us just point out few most important from practical point of view.
\begin{itemize}
\item If possible one should use multiple cores computation to speed up the algorithm.
 By default all  but one cores are used. User can override this with \emph{numb.cores} parameter
\item To avoid algorithm getting stuck in the local minimum one should run it with random initialization multiple times (see parameter \emph{numb.runs}). Default value is 20. 
{We advice to use as many runs as possible (100 or even more).}

\item We recommend doing a hot-start initialization with some non-random segmentation. 
Such a segmentation could be result of some expert knowledge or different clustering method e.g. SSC.  We explore this
  option in simulation studies.
\item  Parameter \emph{max.dim} should reflect how large  dimensions of clusters are expected to be. Default value is 4.
\end{itemize}

\section{Acknowledgements}
M. Bogdan, P. Graczyk, F. Panloup and S. Wilczy\'nski
thank the  AAP MIR 2018-2020 (University of Angers) for its support. P. Graczyk and F. Panloup are grateful  to  SIRIC  ILIAD  program  (supported  by  the  French National  Cancer  Institute  national  (INCa),  the  Ministry  of  Health  and  the  Institute  for  Health  and Medical Research)
 and to PANORisk program of the R\'egion Pays de la Loire
 for their support. F. Panloup  is also supported by the \textit{Institut de Canc\'erologie de l'Ouest}. M. Bogdan was also supported by the Polish National Center of Science via grant 2016/23/B/ST1/00454. M. Staniak was supported by the Polish National Center of Science grant 2020/37/N/ST6/04070.

\section{Appendix. Proof of the PESEL Consistency Theorem } \label{sec:appendix}

In the following we shall denote the  sample covariance matrix
$$S_n = \cfrac{(X - \bar{X}) ^T (X- \bar{X})}{n},$$ 
the
covariance matrix
$$ \Sigma_n = E \left( S_n \right)  = \cfrac{ M_{n\times p}^T M_{n\times p}}{n} + 
\frac{n-1}n \sigma^2  Id$$

and the heterogeneous PESEL function $F(n,k)$
\begin{align}\label{BIC_model_PPCA}
&F(n,k)=   \nonumber \\
&-\frac{n}{2}
 \underbrace{\left[\sum^k_{j=1} \ln(\lambda_j) + (p-k)\ln\left(\frac{1}{p-k}\sum_{j=k+1}^p \lambda_j\right) + p\ln(2\pi)+p\right]}_{G(k)} \nonumber \\
 & -\underbrace{\ln(n)\frac{pk- \frac{k(k+1)}{2} + k + p + 1}{2}}_{P(n,k)} 
\end{align}

\begin{proposition} \label{prop:lil_extended}
Let $E$ have i.i.d. entries with a normal distribution $\N(0,\sigma^2)$.
 There exists a constant  $C>1$ such that almost surely,
 \begin{align*}
 \textrm{ }\exists_{n_0}  \forall_{n\ge n_0} \quad &  \|\frac1n (E-\bar E)^T (E-\bar E)-\sigma^2 Id\|\le \\
 & C\frac{\sqrt{2\ln\ln n}}{\sqrt{n}}\ 
 \end{align*}

\end{proposition}
\begin{proof}
It is a simple corollary of LLN and LIL.  
The term $jk$ of $ \frac1n (E-\bar E)^T (E-\bar E)$
equals
$$
 \frac1n (E_{\bullet j} -\bar E_{\bullet j} {\bf 1})^T (E_{\bullet k}-\bar E_{\bullet k}{\bf 1})= \frac1n \sum_{i=1}^n E_{ij} E_{ik} - \bar E_{\bullet j}\bar E_{\bullet k}.
$$
An upper bound of convergence of $\frac1n \sum_{i=1}^n E_{ij} E_{ik}$
to $\sigma^2 \delta_{jk}$ is $\frac{\sqrt{2\ln\ln n}}{\sqrt{n}}$. It is easy to show that an upper bound of convergence of
$\bar E_{\bullet j}\bar E_{\bullet k}$ to 0 is $ (\frac{\sqrt{2\ln\ln n}}{\sqrt{n}})^2 \le \frac{\sqrt{2\ln\ln n}}{\sqrt{n}}.$
\end{proof}

 \begin{proposition} \label{prop:LILX}
Let $E$ have i.i.d. entries with a normal law $\N(0,\sigma^2)$. There exists a constant  $C>1$ such that almost surely,
 \begin{align}\label{LILX}
 \textrm{ }\exists_{n_0} \ \forall_{n\ge n_0} \quad &  \|\frac1n (X-\bar X)^T (X-\bar X)- ({ L} +\sigma^2 Id)\|\le \nonumber \\
  & C\frac{\sqrt{2\ln\ln n}}{\sqrt{n}}\ 
 \end{align}

\end{proposition}

\begin{proof}

It is easy to check that
$ X-\bar X= M+E- \bar E.$ 
We write
\begin{align*}
\frac1n (X-\bar X)^T (X-\bar X)=
\frac1n  M^T M + \frac1n (E-\bar E)^T (E-\bar E) \\
+ \frac1n (M^TE+E^TM)-\frac1n (M^T\bar E+\bar E^TM)
\end{align*}

To the first two terms we apply, respectively, the hypothesis \eqref{eq:pesel_consistency_M_assumption} and the Proposition \ref{prop:lil_extended}.

To prove the right pace of convergence of the third term $\frac1n (M^TE+E^TM)$ we consider every term $ (M^TE)_{ij}= \langle M_{{\bullet } i}, E_{ {\bullet }j}\rangle$ for which we use a generalized version of Law of Iterated Logarithm from \cite{petrov1995limit}.
Its assumptions are trivially met for random variables $$ M_{li} E_{lj} \sim \N(0, M_{l i}^2 \sigma^2)$$
as they are Gaussian and $\frac{B_{n+1}}{B_n} = \frac{n+1}{n} \to 1$, where $B_n$ is defined as   $B_n = \sum_l M_{li}^2 \sigma^2$.
% = M_{\cdot, i}^T M_{\cdot, i}  \sigma^2 = n \tilde{M}_{i,i} \sigma^2$.
 Then, by \cite{petrov1995limit}, the following holds

$$ \limsup_{n \to \infty} \frac{\sum_{l} M_{l i} E_{lj}} { \sqrt{2 B_n \log \log B_n}}=1 \quad a.s.$$

The fourth term $\frac1n (M^T\bar E+\bar E^TM)$ is treated using Cauchy-Schwarz inequality:

\begin{align*}
|(\frac1n M^T\bar E)_{ij}| & =\frac1n |\langle M_{\bullet i},\bar E_{\bullet j} \rangle| \\
& \le \frac1n  \|   M_{\bullet i}  \| \|  \bar E_{\bullet j}  \| = \frac1n  \|   M_{\bullet i}  \|
\sqrt{n\ \overline{E_{\bullet j}}^2} \\
& = \frac1{\sqrt{n}} \|   M_{\bullet i}  \|
| \overline{E_{\bullet j}}|.
\end{align*}

By LIL, $| \overline{E_{\bullet j}}|\le C \frac{\sqrt{2\ln\ln n}}{\sqrt{n}}$.  The square of the first term $ (\frac1{\sqrt{n}} \|   M_{\bullet i}  \|)^2$ converges to a finite limit by the assumption \eqref{eq:pesel_consistency_M_assumption}.
\end{proof}

%
% LEMMA MAX DIFFERENCE OF EIGENVALUES
%
%
%
%

\begin{lemma}\label{lemmaLILeigen}
There exists $C'>0$ such that \textrm{almost surely},
\begin{align}\label{LILeigen}
 \exists n_0\ \forall n\ge n_0\quad   
\| \lambda(S)- \lambda(\Sigma)  \|_{\infty} \le  C'\frac{\sqrt{2\ln\ln n}}{\sqrt{n}}, 
 \end{align}
 where $\vS$ is sample covariance matrix for data drawn according to model~\eqref{eq:model_for_consistency}, $\Sigma$ is its expected value and function $\lambda(\cdot)$ returns sequence of eigenvalues.
 \end{lemma}
\begin{proof}
 
~\\
\noindent  Observe that 
\begin{align*}
&\norm{ \cfrac{(X - \overline{X}) ^T  (X- \bar X)}{n}  -\Sigma }_{\infty}  \\
& \leq \|\frac1n (X-\bar X)^T (X-\bar X)- ({ L} +\sigma^2 Id)\|  \\
& +\| ({ L} +\sigma^2 Id) - \Sigma \|
\end{align*}

We apply  Proposition \ref{prop:LILX} to the first term and   the assumption \eqref{eq:pesel_consistency_M_assumption} to the second one.

\noindent Inequality \eqref{LILeigen} holds because \eqref{LILX} holds and, by Theorem A.46(A.7.3) from \cite{bai}, when $A,B$ are symmetric, it holds 
\begin{equation*}
\max_k|\lambda_k(A)-\lambda_k(B)| \le \|A-B\|,
\end{equation*}
where function $\lambda_k(\cdot)$ denotes the $k^{\text{th}}$ eigenvalue in the non-increasing order.
\end{proof} 

\begin{proof}[Proof of Theorem \ref{thm:pesel}]

~\\

Let $\epsilon_n= \max^i|\lambda_i(S_n)-\lambda_i(L)|$.
From Lemma \ref{lemmaLILeigen} we have $\lim_n\epsilon_n=0 $ almost surely, so for $k\le k_0-1$, for almost all samplings,  there exists $n_0$ such that if $n\ge n_0,$
\begin{equation*}
\epsilon_n<\sigma^2\ \textrm{and}\ \epsilon_n < \frac14
\min_{k\le k_0-1} c_k(\gamma),
\end{equation*}
where $c_k(\gamma)=\gamma_{k+1}-\frac{\sum_{k+2}^p \gamma_i}{p-k-1}>0$.

\noindent We study the sequence of non-penalty terms $G(k)$ (see \eqref{BIC_model_PPCA}). For  simplicity, from now on, we use notation  $\lambda_j=\lambda_j(S_n)$. We consider $G(k)-G(k+1)$ thus getting rid of the minus sign.

\begin{align*}
G(k) &- G(k+1) = \\
&=\ln\lambda_{k+1}+(p-k-1)\ln \frac{\sum_{k+2}^p \lambda_j}{p-k-1}  \\
& - (p-k)\ln \frac{\sum_{k+1}^p \lambda_j}{p-k}  \\
&=\ln\lambda_{k+1} - \ln \frac{\sum_{k+2}^p \lambda_j}{p-k-1} \\
& + (p-k) \left[\ln \frac{\sum_{k+2}^p \lambda_j}{p-k-1} - \ln \frac{\sum_{k+1}^p \lambda_j}{p-k}\right]
\end{align*}

\noindent  Let us now denote $a=\lambda_{k+1}$ and $b=\frac{\sum_{k+2}^p \lambda_j}{p-k-1}$. Then the above becomes:

\begin{align*}
\ln a - \ln b + (p-k) \left[ \ln b - \ln \frac{b (p-k-1) + a}{p-k}\right]
\end{align*}

{\bf Case $k\le k_0-1.$}

\noindent We will use notation as above and exploit concavity of $\ln$ function by taking Taylor expansion at point $x_0$
%%%%%%%%%%%%%%%%%%%%%%%%%%%%%%%%%%%%%%%%%%%%%%%%%%%%
\begin{equation*}
    f(x) = f(x_0) + f^\prime(x_0) (x-x_0) + \frac{f^{\prime\prime} (x^\star)}{2} (x-x_0)^2,
\end{equation*}
where $x^\star \in (x, x_0).$

\noindent Let $x_0 = \theta x_1 + (1-\theta) x_2$ and $x=x_1$. Then
\begin{align*}
     f(x_1) &= f(x_0) + f^\prime(x_0) (1- \theta) (x_1-x_2) \\
     &+ \frac{f^{\prime\prime} (x^\star_1)}{2} (1- \theta)^2 (x_1-x_2)^2.
\end{align*}

\noindent Similarly, we take $x=x_2$, multiply both
equations by $\theta$ and $1-\theta$ respectively and sum them up. We end up with the formula
%%%%%%%%%%%%%%%%%%%%%%%%%%%%%%%%%%%%%%%%%%%%%%%
\begin{align*}
&    \theta f(x_1) + (1-\theta) f(x_2) =  \\
 &     f(x_0) + \theta (1- \theta) (x_2-x_1)^2 \left[ \frac{f^{\prime\prime} (x^\star_1)}{2} (1- \theta) +  \frac{f^{\prime\prime} (x^\star_2)}{2} \theta \right].
\end{align*}

\noindent In our case $f^{\prime\prime}(x) = -\frac{1}{x^2}$, which means that 
$\frac{f^{\prime\prime} (x^\star_i)}{2} < \frac{f^{\prime\prime} (x_2)}{2}$
because $x^\star_1 \in (x_1, x_0) < x_2$ and  $x^\star_2 \in (x_0, x_2) < x_2$.
This yields
\begin{align}
    &\theta f(x_1) + (1-\theta) f(x_2) - f(x_0) = \nonumber \\ 
    &\theta (1- \theta) (x_2-x_1)^2 \left[ \frac{f^{\prime\prime} (x^\star_1)}{2} (1- \theta) +  \frac{f^{\prime\prime} (x^\star_2)}{2} \theta \right]  \label{eq:inequality_determ_pesel_theta} \\
    &< \theta (1- \theta) (x_2-x_1)^2 \frac{f^{\prime\prime} (x_2)}{2}\nonumber
\end{align}

Now, going back to $G(k)$, we set 
\begin{equation}\label{X1X2}
 x_1 = b = \frac{\sum_{k+2}^p \lambda_j}{p-k-1}, \quad
x_2 =  a = \lambda_{k+1}, \quad
\theta = 1-\frac{1}{p-k}.
\end{equation}

\noindent By multiplying both sides of \eqref{eq:inequality_determ_pesel_theta} by $p-k$ we get

\begin{align*}
    &(p-k-1) \ln \left( \frac{\sum_{k+2}^p \lambda_j}{p-k-1} \right) + \ln (\lambda_{k+1}) \\
    & \; - (p-k) \ln \left( (1-\frac{1}{p-k}) \frac{\sum_{k+2}^p \lambda_j}{p-k-1} + \frac{1}{p-k} \lambda_{k+1} \right) \\
    &< - \left(1- \frac{1}{p-k} \right) \left( \lambda_{k+1}-\frac{\sum_{k+2}^p \lambda_j}{p-k-1} \right)^2 \frac{1}{2 \lambda_{k+1}^2}
\end{align*}

\noindent So, using $k+1 \leq k_0$ in the last inequality, we get
%%%%%%%%%%%%%%%%%%%%%%%%%%%%%%%%%%%%%%%%%%%%%%%%%
\begin{align*}
    G(k+1) & - G(k) > \\
    & \left(1- \frac{1}{p-k} \right) \left( \lambda_{k+1}-\frac{\sum_{k+2}^p \lambda_j}{p-k-1} \right)^2 \frac{1}{2 \lambda_{k+1}^2} \nonumber \\
     &= \frac{p-k-1}{p-k} \left( \lambda_{k+1}-\frac{\sum_{k+2}^p \lambda_j}{p-k-1} \right)^2 \frac{1}{2 \lambda_{k+1}^2} \nonumber \\
     &{>} \; \frac{p-k_0-1}{p-k_0} \left( \lambda_{k+1}-\frac{\sum_{k+2}^p \lambda_j}{p-k-1} \right)^2 \frac{1}{2 \lambda_{1}^2}
\end{align*}

From Lemma \ref{lemmaLILeigen}, $\lambda_i\in [\gamma_i+\sigma^2-\epsilon_n,\gamma_i+\sigma^2+\epsilon_n]$, where  $\epsilon_n$ goes  to 0 and

\begin{align*}
& \left( \lambda_{k+1}-\frac{\sum_{k+2}^p \lambda_j}{p-k-1} \right) \ge \\
& \; \; \ge \gamma_{k+1}+\sigma^2-\epsilon_n -\frac{\sum_{k+2}^p (\gamma_i+\sigma^2+\epsilon_n)}{p-k-1} \\
& \; \; = c_k(\gamma) - 2\epsilon_n\ge \min_{k\le k_0-1} c_k(\gamma) - 2\epsilon_n>0
\end{align*}

for some constants $c_k(\gamma)$. Thus
\begin{align*}
    G&(k+1) - G(k) \\
    & > \frac{p-k_0-1}{p-k_0} (\min_{k\le k_0-1} c_k(\gamma) - 2\epsilon_n)^2\frac{1}{2(\gamma_1+\sigma^2+\epsilon_n)^2} \\
    & >\frac{C'}{2}\min_{k\le k_0-1} c_k(\gamma)>C>0
\end{align*}
where $C,C'$ are constants independent of $k$ and $n$.
It follows that for $n$
 large enough 
\begin{align*}
\frac{n}{2}[G(k+1)-G(k)] & \\
& \ge \frac{n}{2} C \\
& \gg \frac{\ln n}{2} (p-k) \\
& = P(n,k+1)-P(n,k).
\end{align*}

This implies that the PESEL  function $F(n,k)=
\frac{n}{2}G(k)-P(n,k)$ is strictly increasing for $k\le k_0$.
\\

%%%%%%%%%%%%%%%%%%%%%%%%%%%%%%%%%%%%%%%%%%%%%%%

{\bf Case $k\ge k_0.$}
By Lemma  \ref{lemmaLILeigen} we have that, for almost all samplings,  there exists $n_0$  such that if $n\ge n_0,$
\begin{align*}
\epsilon_n \le  C\frac{\sqrt{2\ln\ln n}}{\sqrt{n}}\ 
\textrm{and }  \epsilon_n<\frac12\sigma^2.
 \end{align*}

 We apply the formula \eqref{eq:inequality_determ_pesel_theta} and as before, we use the notations \eqref{X1X2}.
It yields
 \begin{align*}
G(k+1)-G(k)&\le \left( 1-\frac{1}{p-k}\right)\left(
\lambda_{k+1}-\frac{\sum_{k+2}^p \lambda_j}{p-k-1} \right)^2 \frac{1}{2b^2} \\
& \le (\lambda_{k+1}-b)^2 \frac{1}{2b^2}\\
&\le  (|\lambda_{k+1}-\sigma^2|+ |\sigma^2-b|)^2 \frac{1}{2b^2}
\\
 &\le  (|\lambda_{k+1}-\sigma^2|+ \frac{\sum_{k+2}^p |\sigma^2-\lambda_j|}{p-k-1})^2 \frac{1}{2b^2}\\
 &\le 4\epsilon_n^2 \frac{1}{2(\sigma^2-\epsilon_n)^2}
 \le  C^2\frac{{2\ln\ln n}}{{n}} \frac{4}{2\sigma^4} \\
& =C'\frac{{\ln\ln n}}{{n}}
 \end{align*}
 and consequently
 $$
 \frac{n}{2} \left[ G(k+1)-G(k) \right] \le C'' \ln\ln n
 $$
 
 \noindent Recall that the PESEL function equals $F(n,k)=
\frac{n}{2}G(k)-P(n,k)$.  
 The increase of $\frac{n}{2}G(k)$ is smaller than the rate $\ln\ln n$, while the increase of penalty  
 $P(n,k+1)-P(n,k)=\frac{\ln n}{2} (p-k)$ is of rate $\ln n$. Consequently, there exists
$n_1$ such that for $n > n_1$, the PESEL function is strictly decreasing for $k\ge k_0$
with probability 1.

\noindent We saw in the first part of the proof that
 the PESEL  function $F(n,k)$ is strictly increasing for $k\le k_0$, for $n$ big enough.
 It implies that with probability 1, there exists $n_2$ such that for $n>n_2$ we have $\hat k_0(n)= k_0$.
\end{proof}


\bibliographystyle{plain}      % mathematics and physical sciences
\bibliography{bibliography}

\end{document}